\documentclass[journal, 10pt]{IEEEtran}

\usepackage{amssymb}
\usepackage{relsize}

\newcommand{\matrixnorm}[1]{{\left\vert\kern-0.25ex\left\vert\kern-0.25ex\left\vert #1 \right\vert\kern-0.25ex\right\vert\kern-0.25ex\right\vert}}
\usepackage[pdftex]{graphicx}
\usepackage{hyperref}
\usepackage{cite}
\usepackage{multirow}
\usepackage{flushend}
\ifCLASSINFOpdf
\else
\fi
\usepackage{amsmath}
\DeclareMathOperator*{\argmin}{arg\,min}
\usepackage{algorithm}
\usepackage{algpseudocode}
\usepackage{algorithmicx}
\usepackage{color}
\usepackage{amsthm}
\usepackage{tikz}
\usepackage{enumitem}
\usepackage{cases}
\usetikzlibrary{shapes}

\usepackage{url}
\algrenewcommand\algorithmicrequire{\textbf{Input:}}
\algrenewcommand\algorithmicensure{\textbf{Output:}}

\newtheorem{theorem}{Theorem}
\newtheorem{cor}{Corollary}
\newtheorem{lemma}{Lemma}

\begin{document}
\bstctlcite{IEEEexample:BSTcontrol}

\title{A Probably Approximately Correct Analysis of Group Testing Algorithms}
\author{Sameera~Bharadwaja~H.
	and~Chandra~R.~Murthy
\thanks{Sameera Bharadwaja H. and Chandra R. Murthy are with the Dept. of ECE, Indian Institute of Science, Bangalore, India (email: sameerah@iisc.ac.in; cmurthy@iisc.ac.in). A part of this work was published in~\cite{Sameera_Chandra_ISIT_2022}.}\vspace{-0.7cm}
}

\maketitle

\begin{abstract}
We consider the problem of identifying the defectives from a population of items via a \emph{non-adaptive group testing} framework with a random pooling-matrix design. We analyze the sufficient number of tests needed for \emph{approximate set identification}, i.e., for identifying \emph{almost} all the defective and non-defective items with high \emph{confidence}. To this end, we view the group testing problem as a function learning problem and develop our analysis using the probably approximately correct (PAC) framework. Using this formulation, we derive sufficiency bounds on the number of tests for three popular binary group testing algorithms: column matching, combinatorial basis pursuit, and definite defectives. We compare the derived bounds with the existing ones in the literature for exact recovery theoretically and using simulations. Finally, we contrast the three group testing algorithms under consideration in terms of the sufficient testing rate \emph{surface} and the sufficient number of tests \emph{contours} across the range of the approximation and confidence levels.
\end{abstract}


\IEEEpeerreviewmaketitle

\section{Introduction}
\label{sec:Introduction}

\IEEEPARstart{I}{dentifying} a set of $k$ \emph{defectives} from a population of $n$ \emph{items} is an interesting problem. A n\"{a}ive solution is to test the items individually, which requires $n$ tests. This approach is inefficient if $n$ is very large or there are constraints on time-to-test, cost budget, or testing hardware and resource constraints.

An alternative is to pool items together and run $m < n$ tests in parallel, with each test pooling random subsets of the items together. These are collectively called group testing algorithms (or pool testing algorithms) and work as follows. The outcome of a group test is negative if and only if none of the defective items participate in that group test; it is positive otherwise. When $k \ll n$, which holds in many applications like identifying a rare disease from a set of blood samples, testing a population for an infection in the early stage of an epidemic, identifying defective industrial products in a high-yield production line, etc., this approach can help significantly reduce the sufficient number of tests needed.

Group testing was first introduced by Dorfman in $1943$ during world war II to test the prospective entrants into the US military for Syphilis~\cite{Dorfman_1943}. Owing to its success, group testing has been adopted in diverse applications ranging from (other) infectious disease detection \cite{Ajay_Bharti_2009, Cohen_Nir_Eldar_2021} to multiple access control protocols \cite{Wolf_1985}, cognitive radios \cite{Sharma_Murthy_2014}, drug and DNA library scanning \cite{Hwang_Du_2006, Ngo_Du_1999, Raghu_Peter_2008, Xin_2009}, product testing \cite{Sobel_Groll_1959}, etc. Group testing has gained renewed interest recently due to global developments related to COVID-19~\cite{Yelin_Aharony_2020, Sameera_Chandra_TSP_2022}. For example, it is possibly applicable to monkeypox detection~\cite{Li_Zhao_Kimberly_2010} and to scale-up testing to detect Khosta 2, a type of sarbecovirus~\cite{Seifert_Bai_2022}.

Group testing methods can be categorized into \emph{adaptive} and \emph{non-adaptive} types~\cite{Ding_Zhu_1999}. In adaptive group testing, the tests are performed in multiple stages wherein the group design in the current stage depends on the test outcomes from the previous stage. Dorfman-style testing falls under this category. In contrast, in the non-adaptive method, all the required tests are performed in a single stage, followed by the application of a suitable decoding algorithm~\cite{Aldridge_Balsassini_2014, Chan_Jaggi_Saligrama_2011, Chan_Jaggi_Saligrama_Agnihotri_2014} to recover the individual \emph{item status} (i.e., defective or non-defective).

In this paper, we focus on random pooling-based non-adaptive group testing~\cite{Sharma_Murthy_2014}. Here, the decision about which items will participate in which group test is predetermined and is encoded in a random binary \emph{test matrix} denoted by $\mathbf{A} \in \{0, 1\}^{m \times n}$~\cite{Aldridge_GT_IT_2019}. The $(i, j)$-th element of $\mathbf{A}$, denoted by $a_{ij}$, takes the value $1$ or $0$ depending on whether the $j$-th item participates in the $i$-th group test or not, respectively. The \emph{item vector} is denoted by $\mathbf{x} \in \{0, 1\}^n$, whose $j$-th entry, denoted by $x_j$, takes the value $1$ if the $j$-th item is defective, and $0$ otherwise. In addition, the support of $\mathbf{x}$ is denoted by the set $\mathcal{K}$ such that $|\mathcal{K}| = k$, where $|\cdot|$ denotes the cardinality of a set. Finally, the \emph{outcome} of the $i$th group test is $y_i = \vee_{j=1}^{n}~a_{ij}x_j,~i \in [m]$, where $\vee$ denotes the boolean OR-ing operation. Thus, the outcome is $1$ if the group test includes one or more defective items and is $0$ otherwise.

Once the group test outcomes ($y_i, i \in [m]$) are collected, a decoding algorithm aims to output an item vector estimate denoted by $\hat{\mathbf{x}} \in \{0,1\}^n$. The estimate of the defective item set, $\hat{\mathcal{K}}$ is then simply the support of $\hat{\mathbf{x}}$, i.e., the set of indices corresponding to the nonzero entries in $\hat{\mathbf{x}}$. Some of the well-known binary group testing (decoding) algorithms studied in this work include~\cite[Chapter 2]{Aldridge_Balsassini_2014}~\cite{Chan_Jaggi_Saligrama_Agnihotri_2014}
\begin{enumerate}
	\item Column Matching (CoMa), also called Combinatorial Orthogonal Matching Pursuit (COMP).
	\item Combinatorial Basis Pursuit (CBP), also called Coupon Collector (CoCo) algorithm.
	\item The Definite Defectives (DD) algorithm.
\end{enumerate}

Note that, under the group testing model described above, CoMa and CBP only make false positive errors, i.e., $\hat{\mathcal{K}} \supseteq {\mathcal{K}}$, while DD only makes false negative errors, i.e., $\hat{\mathcal{K}} \subseteq {\mathcal{K}}$~\cite{Aldridge_Balsassini_2014}. We analyze these separately in the sequel.

Now, when \emph{exact set identification} is considered, the decoding algorithm succeeds when $\hat{\mathcal{K}} = \mathcal{K}$ and fails otherwise. Let $\mathbb{P}_{\mathbf{A}}\left(\hat{\mathcal{K}} \neq \mathcal{K}\right)$ denote the probability of error given the set of defectives, $\mathcal{K}$. Then, the average probability of error under a combinatorial setting can be written as~\cite{Aldridge_GT_IT_2019}
\begin{equation}
	\label{eq:Error_Defective_Set_Estimate}
	\mathbb{P}(\text{err}) = \mathbb{P}_{\mathbf{A},\mathcal{K}}(\hat{\mathcal{K}} \neq \mathcal{K}) = \frac{1}{{n \choose k}}\sum_{\mathcal{K}: |\mathcal{K}| = k}\mathbb{P}_{\mathbf{A}}(\hat{\mathcal{K}} \neq \mathcal{K}).
\end{equation}
Characterizing $\mathbb{P}(\text{err})$ as a function of $k$, $n$, and $m$ is helpful in using group testing for  practical applications. For instance, one can ask: \emph{Given $k$ and $n$ along with an upper bound on $\mathbb{P}(\text{err})$, what is the sufficient number of tests, $m_S$, required by a group testing algorithm?} In the following subsections, we briefly discuss the existing work toward answering this question before presenting our main contributions. 
\vspace{-0.4cm}
\subsection{Prior Work}
\label{sec:Prior Work}
Under $\delta^\prime \equiv \mathbb{P}(\text{err}) = n^{-\delta}, \delta > 0$ and $k = o(n)$, the authors in~\cite{Chan_Jaggi_Saligrama_2011, Chan_Jaggi_Saligrama_Agnihotri_2014} show that an upper bound on the number of tests is $ek(1+\delta)\log n$ and $2ek(1+\delta)\log n$ for CoMa and CBP algorithms, respectively.\footnote{$\log$ denotes natural logarithm unless specified otherwise.} Further, the linear programming (LiPo) decoder requires no more than $O(k\log n)$ tests under the same conditions with the constant factor associated with the asymptotic expression being a function of $(1\!+\!1/k)$, $\log k/ \log n$ and $\delta$~\cite{Chan_Jaggi_Saligrama_Agnihotri_2014}. Similarly, the sufficient number of tests for DD and sequential COMP (SCOMP) algorithms is also $C_ak\log n,$ with a decoder-dependent constant $C_a > 0$~\cite{Aldridge_Balsassini_2014}. In addition, a lower bound on the required number of tests is $(1-\delta^\prime)k\log(n/k)$~(See \cite{Chan_Jaggi_Saligrama_Agnihotri_2014} and references therein, e.g.,~\cite{Malyutov_1978, Arkadii_2004}.) \textcolor{black}{In the above works, the entries of the test matrix are chosen i.i.d. from the $\text{Bernoulli}(p)$ distribution, where $p \in (0,1)$. Further, the author in~\cite{Aldridge_2017} presents an improved converse bound for this Bernoulli-design group testing.}

\textcolor{black}{Other test matrix designs have been considered and contrasted with the Bernoulli design. The authors in~\cite{Coja-Oghlan_Oliver_2020} draw inspiration from spatially-coupled low-density parity check (LDPC) codes for designing their test matrix. For any $0\!<\!\beta\!<\!1$, $k\!\sim\!n^\beta$, $\epsilon^\prime\!>\!0$, the authors derive lower and upper bounds on the number of tests as $(1\!-\!\epsilon^\prime)m_{\text{inf}}$ and $(1\!+\!\epsilon^\prime)m_{\text{inf}}$, respectively, where $m_{\text{inf}} = \max\{\beta/\log^2 2, (1\!-\!\beta)/\log 2\} k\log n$. The authors in~\cite{Coja-Oghlan_Gebhard_Max_2020} show similar results for the number of tests below which any group testing algorithm fails and above which the SCOMP and/or DD algorithms succeed.}

\textcolor{black}{In~\cite{Aldridge_Johnson_Scarlett_2016, Johnson_Aldridge_Scarlett_2019}, the authors show that a (near-)constant tests-per-item design with $m(\log 2)/k$ tests-per-item (chosen with replacement) requires $23.4\%$ fewer tests (correspondingly, the testing rate improves by $\approx 30\%$) than the Bernoulli design, when coupled with COMP or DD algorithms. The authors in~\cite{Tan_Scarlett_2021} show that, in the sub-linear sparsity regime (i.e., $k = \Theta(n^\beta),~\beta \in (0,1)$), the DD algorithm under a constrained design with at most a fixed number of tests-per-item and $\rho = O\left((n/k)^\beta\right),~\beta \in [0,1)$ items-per-test, yields an improved achievability result compared to the COMP algorithm under an unconstrained design. The work in~\cite{Gandikota_Jaggi_2019} presents an analysis of constrained design (e.g., constraints on the number of tests per item or the \emph{item-divisibility} and on the number of items per test) and shows that even a small \emph{amount of constraint} can have a significant effect on the information-theoretic bound. Recently, the authors in~\cite{Oliver_Max_Olaf_2022} considered both a constant column-weight design (termed $\Delta$-divisible item setting) and a constant row-weight design (termed $\Gamma$-sized test setting) while analyzing the achievable number of tests of various decoding algorithms.}

\textcolor{black}{In contrast to the sub-linear regime, the works in~\cite{Malyutov_1978, Malyutov_2013} consider very sparse regime, where $k = O(1)$.} The authors in~\cite{Scarlett_Cevher_2017} present a sufficiency condition on $m$ for a weakened version of the maximum likelihood decoder in the non-asymptotic regime, i.e., when the problem dimensions are finite. The non-asymptotic bounds are more practically useful: for instance, in a typical RT-qPCR based Covid-19 testing kit, one can accommodate $\sim 96$ or $384$ (pooled) samples~\cite{Sameera_Chandra_TSP_2022}. However, the analysis of the maximum likelihood-based decoder is dependent on the decoder having the knowledge of $k$ and is typically computationally intractable even for moderate sized problems~\cite{Aldridge_GT_IT_2019, Johnson_Aldridge_Scarlett_2019}. Therefore, practical group testing algorithms like CoMa, CBP, and DD are more attractive.
\subsection{Approximate Defective Set Identification}
\label{sec:Approximate Defective Set Identification}
In many group testing applications,  \emph{approximate} defective set recovery is sufficient~\cite{Aldridge_GT_IT_2019}. Here, the estimated defective set can contain missed defective items or false positive items. \textcolor{black}{In this direction, the authors in~\cite{Atia_Saligrama_2009} present an achievability bound on the number of tests when up to $\alpha k$ misses out of $k$ defectives are allowed given $n$ items, as $O\left(e^{(1-\alpha)k \log n}/H(e^{-\alpha})\right)$, where $H(\cdot)$ denotes the binary entropy.} The authors in \cite{Scarlett_Cevher_2016, Scarlett_Cevher_2017} consider approximate defect set identification and derive both achievability and converse bounds on $m$ in the \emph{sub-linear regime}, i.e., $k\!=\!\Theta(n^\beta),~\beta\!\in\!(0,1)$. The proof uses information spectrum methods and thresholding techniques from channel coding theory. For achievability, the authors use a maximum likelihood-like recovery algorithm. Similarly, the authors in \cite{Scarlett_Cevher_List_2017} consider a list-decoding algorithm for approximate recovery. The results in all these scenarios show that allowing $\lfloor \alpha k \rfloor$ defectives to be missed relaxes the converse bound on $m$ by at most a multiplicative factor $1\!-\!\alpha$, where $\alpha \in (0, 1)$.

The authors in~\cite{Truong_Aldridge_Scarlett_2020} show that the probability that the maximum of the false positive (FP) or false negative (FN) errors incurred is no more than $\alpha k, \alpha \in (0,1)$ approaches one when $m \geq (1\!+\!\alpha^\prime)(1\!-\!\alpha)\log_2{n \choose k}$ but the same probability approaches zero when $m < (1\!-\!\alpha^\prime)\log_2{n \choose k}$, for arbitrary $\alpha^\prime > 0$, as $n \to \infty$. A treatment of this \emph{all-or-nothing} phenomenon is also presented in~\cite{Jonathan_Ilias_2023}. The authors in~\cite{Lee_Chandrasekher_2019} use sparse graph codes and present a decoder called SAFFRON which can recover $(1\!-\!\epsilon^\prime)k$ defectives with probability $1\!-\!k/n^r$ with $2(1\!+\!r)C^\prime k\log_2 n$ tests, where $C^\prime$ is a function of $\epsilon^\prime > 0$ and $r \in \mathbb{Z}^+$.

The available bounds on the sufficient number of tests for approximate defective set recovery also depend on the underlying test matrix design, i.e., the distribution from which the test matrix is drawn. For example, in the sub-linear regime, with an i.i.d. $\text{Bernoulli}(\log 2/k)$ test matrix, $k\log(n/k)/\log^22$ tests are sufficient for the error probability (appropriately defined to account for the maximum number of FPs and FNs allowed) to approach zero \cite{Scarlett_Cevher_Separate_2018}. Similarly, under a doubly-regular design with column-weight $r$ and row-weight $s$, it is known that the FNR is minimized for larger $r$ and smaller $s$, whereas the number of tests is $rn/s$~\cite{Tan_Scarlett_2023}. Therefore, they fix different (nonzero) values for the allowed FNR and numerically analyze the number of tests. They define $\text{FPR}=\mathbb{E}[|\mathcal{\hat{K}} \setminus \mathcal{K}|]/|[n]\setminus\mathcal{K}|$ and $\text{FNR}=\mathbb{E}[|\mathcal{K} \setminus \mathcal{\hat{K}}|]/|\mathcal{K}|$, where $[n] \triangleq \{1, \ldots, n\}$ as a measure of the average approximation errors. In contrast, we are interested in the number of FP or FN errors allowed at a given $n$, $k$ and with the confidence parameter, $\delta$. In essence, we derive a lower bound on the cumulative distribution function of the approximation errors.

Recently, the authors in~\cite{Iliopoulos_Ilias_2021} addressed a question related to the \emph{computational-statistical gap} (CSG) in the non-adaptive group testing paradigm, where one is interested in $(1\!-\!o(1))-$approximate recovery in the $k = o(n)$ regime, for the $\text{Bernoulli}(\log 2/k)$ test design~\cite[Sec. IIA]{Scarlett_Cevher_Separate_2018}. The CSG is defined as the gap between the number of tests above which recovery is information-theoretically possible and the number of tests required by the currently best-known efficient algorithms to succeed~\cite{Iliopoulos_Ilias_2021}. The authors provide evidence that the gap can be closed when $m \geq (1\!+\!\alpha)k\log_2(n/k),~\alpha\!\in\!(0,1)$ and $n, k \to \infty$, enabling $(1\!-\!\alpha^\prime)m$ tests to contain at least $\lfloor (1\!-\!\alpha)k \rfloor$ defectives asymptotically almost surely, for any $\alpha^\prime > 0$. The authors show that the absence of CSG implies that a practical local search routine succeeds in solving the smallest satisfying set (SSS) estimator under the said regime, which otherwise has a combinatorial complexity.

In contrast to the usual error probability or the FP/FN rates, the author in~\cite{Aldridge_2021} considers a new metric called the expected number of tests per \emph{infected} individual found (ETI). One of the questions posed in this work involves tolerating partial recovery and analysis of the SAFFRON algorithm's ETI, which is shown to be $2e\log_2 (n/k)$. Lastly,~\cite[Section 5.1]{Aldridge_GT_IT_2019, Malyutov_2013} presents bounds on $m$ for CoMa and DD under partial recovery conditions, albeit without proof. For instance, for $m \geq (1+\eta)ek\log(n/k)$, where $\eta > 0$, the average number of FPs output by the CoMa algorithm behaves as $o(k)$, and, therefore, the probability of getting more than $\gamma k$ false positives, for a fixed $\gamma \in (0,1)$, tends to zero. A similar result is presented in the context of the DD algorithm for FN errors. The authors in~\cite{Gilbert_Wen_Strauss_2008} also briefly discuss partial defective set recovery in group testing.

In practical settings with finite resources, one is often interested in ensuring that the probability of the error incurred by a function learned using a finite number of randomly drawn samples exceeding a threshold remains below a small number, called \emph{confidence}. We use the probably approximately correct (PAC) formulation~\cite{Mohri_Afshin_Talwalkar_2012, Shai_Shai_2014} in this work to bridge this gap.

\textcolor{black}{We have presented some of the main results from the vast field of research on group testing relevant to this work. A collection of extensive results and deeper discussions pertaining to various decoder rules, testing protocols, and designs, including the bounds on the number of tests under different constraints, reliability criteria, various measurement and mixing models can be found in~\cite{Aldridge_GT_IT_2019} and references therein. With this background, we now summarize the motivation and contributions of this work.}

\vspace{-0.3cm}
\subsection{Motivation and Contributions}
\label{sec:Motivation for the Current Work}
\textcolor{black}{To the best of our knowledge, a rigorous treatment of sufficiency bounds on $m$ for \emph{approximate} recovery using practical algorithms like CoMa, CBP, and DD accounting for the \emph{randomness} in the test matrix and at nonzero confidence levels is not available in the literature.} This work addresses this gap by viewing these algorithms through the lens of function learning and PAC analysis. In turn, this allows us to shed light on the relationship between PAC-learning and the exact recovery bounds available in the literature~\cite{Chan_Jaggi_Saligrama_Agnihotri_2014, Aldridge_Balsassini_2014}. \textcolor{black}{A fundamental difference between PAC learning~\cite{Mohri_Afshin_Talwalkar_2012, Shai_Shai_2014} and our problem is as follows. In group testing, we can \emph{choose} the data distribution from which the samples are drawn for function learning based on our knowledge of the hypothesis space from which the target function is to be learned (in the context of group testing, function learning corresponds to identifying defective items.) For example, in i.i.d.\ Bernoulli test matrix designs, we can choose the probability $p$ with which a given item participates in a given test. As we will see, this additional degree of freedom allows us to optimize the PAC learning process and obtain tighter and more general bounds on the sufficient number of group tests.}

\textcolor{black}{It is worth mentioning that viewing group testing as a function learning problem has other potential applications: blind chemistry, where one is interested in determining which $k$ out of $n$ reactants in a chemical reaction has the potential to create a particular (useful) detectable compound~\cite{Gilbert_Wen_Strauss_2008};  identification of key design variables for improving  an automobile's fuel efficiency; key-species identification in a complex biological ecosystem; reactions of bacteria in gut micro-biome to a given drug, etc~\cite{Pant_Maiti_2022, Gilbert_Wen_Strauss_2008}.}
\textcolor{black}{Also, multi-label classification with a large number of the number of class labels  $n \sim 10^3\!-\!10^6$ and a small number of classes per input (e.g., an image) is $k\!\ll\!n$ can be performed by a set of $n$ \emph{one-vs.-rest} classifiers. The authors in~\cite{Ubaru_Mazumdar_2017} propose to use a $(k,e)$-disjunct matrix and encode the $n$-dimensional label vector using an $m$-dimensional vector, where $m\!\ll\!n$. The $m$-length binary vector is constructed by learning $m$ one-vs.-rest classifiers. Then, the recovery of the original $n$-length label vector from the $m$-length binary vector is solved as a classical group testing problem. The test design in~\cite{Ubaru_Mazumdar_2017} can tolerate up to $\lfloor e/2 \rfloor$ misclassifications. On a similar note, the authors in~\cite{Malioutov_Varshney_2013} show how group testing can be used in a binary classification problem by posing it as an exact rule-learning problem.}

\textcolor{black}{We note that the PAC formulation is one of the key tools used to analyze machine learning algorithms~\cite{Mohri_Afshin_Talwalkar_2012}. \textcolor{black}{For example, PAC analysis aids in developing a lower bound on the probability that the error rate of the above-mentioned multi-label classifier lies below a certain error threshold.} However, the existing analysis of group testing algorithms does not conform to this notion of characterizing the distribution of the error rates over the randomness of the test matrix. By formulating group testing as a PAC learning problem and developing corresponding achievability bounds, we are able to bridge this gap. Furthermore, since the PAC formulation allows one to characterize the decoding error rate \emph{and} the confidence in a unified framework, it is well suited for applications where the target error tolerance is small but nonzero, as in the examples discussed above.} 

In this paper, we consider random pooling-based non-adaptive group testing and characterize the sufficient number of tests, $m$, required for \emph{approximate} recovery of the defective set with \emph{high-probability} using the PAC learning framework. In doing so, we are able to perform a finer analysis of random pooling-based group testing algorithms by separately accounting for the randomness in the test matrix $\mathbf{A}$ and the defective set $\mathbf{x}$ for three popular group testing recovery algorithms. As we will show, the exact recovery bounds for group testing in the literature are a special case of the PAC analysis-based results.

Our main contributions are as follows:
\begin{enumerate}
	\item We reformulate the defective set identification problem in non-adaptive group testing with random pooling as a function learning problem. Doing so enables us to apply the PAC framework for deriving a sufficiency bound on the number of tests in both exact and approximate recovery conditions for three well-known and popular binary group testing algorithms: CoMa, CBP, and DD. \textcolor{black}{In contrast to existing works, we optimize the design parameters to get a tighter bound on the sufficient number of tests in the CBP algorithm. Further, the existing results on the coupon collector problem do not apply to approximate recovery. Therefore, we also extend the analysis of the coupon collector problem to handle \emph{collection of only a subset of coupons}. The new results are then used to develop measurement bounds for approximate recovery of the defective set using the CBP algorithm.}
	\item \textcolor{black}{The PAC-based analysis is characterized by two parameters: 1) the approximation error tolerance, $\epsilon \in [0,1]$, and 2) the confidence, $1-\delta \in [0,1)$. We present sufficiency bounds on the number of tests required for CoMa, CBP, and DD algorithms as a function of $n$, $k$, $\epsilon$, and $\delta$. Thus, we develop a common framework to arrive at sufficiency bounds on the number of tests for both exact and approximate set recovery with high probability. As a result, the exact defective set recovery results found in the literature emerge as a special of our analysis.}
	\item \textcolor{black}{We derive the order-wise behavior of the PAC bounds for large $n$ and $k$.  When we fix $n$, $k$, and $\epsilon$, the sufficient number of tests obtained by the PAC analysis is $\propto \log(C_d/\delta)$, for $\delta = (0,1]$ with a constant, $C_d = 2$ for the CBP algorithm and $C_d = 1$ for CoMa and DD algorithms. Further, the sufficient number of tests is $\propto (\log(1/\epsilon) + 1/\epsilon)$ for $\epsilon \in [0,1]$ when we fix $n$, $k$ and~$\delta$.}
	\item \textcolor{black}{We relate the number of false positives (similarly, the number of false negatives) to the approximation error probability $\epsilon$ used in the PAC formulation. This way of relating the theoretical results to the practical metrics of interest makes the results appealing to practitioners also.}
	\item The PAC-based analysis allows us to trade-off the \emph{accuracy} of defective set recovery with the \emph{confidence} with which the decoded set meets that accuracy. We present a visualization of this trade-off in the form of a sufficient number of tests \emph{contour/surface}, which shows its dependence on the approximation error tolerance and the probability of failure to meet the required error tolerance. \textcolor{black}{In summary, PAC analysis enables one to characterize a lower bound on the cumulative distribution of the approximation errors.}
\end{enumerate}

One of the main takeaways from our work is that using the PAC framework for analyzing practical group testing algorithms opens up means to accommodate both exact and approximate recovery \emph{and} account for the randomness in the test matrix, under the same umbrella. Furthermore, we work mainly in the non-asymptotic regime, which is of practical interest, as mentioned in Sec.~\ref{sec:Prior Work}. In addition, the sensitivity of the group testing algorithms to both confidence level and error tolerance can be separately quantified.

\textcolor{black}{The main results are stated formally in Theorems~\ref{thm:CoMa_Bound},~\ref{thm:CBP_s_Bound}, and~\ref{thm:DD_Bound} in the sequel for the CoMa, CBP, and DD algorithms, respectively. Here, provide a brief overview of what will follow. Let $c \triangleq c(a,b)$ denote that the parameter/variable $c$ is a function of two parameters/variables, $a$ and $b$. We begin with the CoMa bound presented in Theorem~\ref{thm:CoMa_Bound}:\\}
\textcolor{black}{\emph{When each entry of the test matrix is chosen i.i.d. from a $\text{Bernoulli}(p)$ distribution with $p \in (0,1)$, the CoMa algorithm succeeds in determining the defective set with approximation error at most $\epsilon$ with confidence $1-\delta$ if the number of tests is at~least
\begin{equation}
	\frac{\log{n-k \choose g_\epsilon + 1}+\log\frac{1}{\delta}}{\log\left(1/(1-(1-p)^k+(1-p)^{g_\epsilon + k + 1})\right)} \nonumber
\end{equation}
with $g_\epsilon \triangleq g_\epsilon(\epsilon, k, p) \in \mathbb{Z}_+ \cup \{0\}$ denoting the allowed number of false positives.} Further, $g_\epsilon = ke\epsilon$ order-wise, for large $k$ and with $p = 1/k$, which leads to the order-wise sufficient number of tests given by 
\begin{equation}
	2ke\left[\log\left(\frac{n}{ke\epsilon+1}\right) + 1 + \frac{\log\left(\frac{1}{\delta}\right)}{ke\epsilon+1}\right] \nonumber
\end{equation}
for large $n$ and $k$ with $n \gg k$. In the exact recovery case, i.e., when $\epsilon = 0$, we get $2ke\left[\log n + 1 + \log\left(\frac{1}{\delta}\right)\right]$. Finally, we note that the above result for the CoMa algorithm corresponds to the sufficient number of tests being $\propto \log(1/\delta)$ and $\propto \log(1/\epsilon) + 1/\epsilon$.
}

\textcolor{black}{We present the following CBP bound in Theorem~\ref{thm:CBP_s_Bound}:\\}
\textcolor{black}{\emph{For a non-adaptive random test design where each test contains $s \triangleq s(n, k)$ (out of $n$) items chosen independently with replacement, the CBP algorithm succeeds in determining the defective set with approximation error at most $\epsilon$ with confidence $1-\delta$ if the number of tests is at~least
\begin{equation}
	\frac{(n-k)}{(1-\eta)s\left(\frac{n-k}{n}\right)^s}\left[ \frac{\log\left(\frac{1}{c\delta}\right)}{g_\epsilon+1} + \frac{g_\epsilon}{g_\epsilon+1} + \log\left(\frac{n-k}{g_\epsilon+1}\right)\right] \nonumber
\end{equation}
with $g_\epsilon \triangleq g_\epsilon(\epsilon, k, s) \in \mathbb{Z}_+ \cup \{0\}$ denoting the allowed number of false positives and a Chernoff parameter, $\eta \triangleq \eta(n, k, g_\epsilon, \delta, s, c) \in (0,1)$ for a constant $c \in (0,1)$.}\\
The CBP analysis is supported by Lemma~\ref{lem:CCP_Lemma}, an extension to the \emph{coupon collector problem} for partial coupon-set collection. For large $n$ and $k$, order-wise, $g_\epsilon = ke\epsilon$, similar to the CoMa case, with $s = 1/\log(n/(n-k))$. Therefore, the sufficient number of tests for the CBP algorithm can be written as
\begin{equation}
	2ke\left[\log\left(\frac{n}{ke\epsilon\!+\!1}\right) \!+\!1\!+\!\log\left(\frac{2}{\delta}\right)\!\left[\frac{1}{ke\epsilon\!+\!1}\!+\!\frac{e}{2k}\right]\right] \nonumber
\end{equation}
for large $n$ and $k$. Similar to the CoMa analysis, we see that the bound $\propto \log(1/\delta)$ and $\propto \log(1/\epsilon) + 1/\epsilon$. In the exact recovery case, we get $2ke\left[\log n \!+\!1\!+\!\left(1\!+\!\frac{e}{2k}\right)\log\left(\frac{2}{\delta}\right)\right]$.
}

\textcolor{black}{The third result in this paper is the DD bound, in the form of an implicit equation, presented as Theorem~\ref{thm:DD_Bound}:\\}
\textcolor{black}{\emph{When each entry of the test matrix is chosen i.i.d. from a $\text{Bernoulli}(p)$ distribution with $p \in (0,1)$, the DD algorithm succeeds in determining the defective set with approximation error at most $\epsilon$ with confidence $1-\delta$ if the number of tests, $m$, satisfies the following implicit equation
\begin{equation}
	{k \choose {d_\epsilon+1}}(1-(d_\epsilon+1) p(1-p)^{k-1+\bar{g}+\tilde{g}})^{m} \leq \delta \nonumber
\end{equation}
with $d_\epsilon \triangleq d_\epsilon(\epsilon, k, p) \in \mathbb{Z}_+ \cup \{0\}$ denoting the allowed number of false negatives, $\bar{g} \triangleq \bar{g}(n, k, p, m)$ and a constant, $\tilde{g} \geq 0$.}
Finally, as $n$ and $k$ grow large, $d_\epsilon = ke\epsilon$ order-wise, with $p = 1/k$. Therefore, the sufficient number of tests for the DD algorithm can be written as 
\begin{equation}
	ke\left[\log\left(\frac{n}{k}\right)\!+\!\frac{\log\left(\frac{1}{\delta}\right)}{(ke\epsilon\!+\!1)\log\left(\frac{n}{k}\right)}\!+\!\frac{\log\left(\frac{ke}{ke\epsilon+1}\right)}{\log\left(\frac{n}{k}\right)}\right] \nonumber
\end{equation}
for large $n$ and $k$. In the exact recovery case with $d_\epsilon = 0$, we get $ke\left(\log (n/k) + \frac{1}{\log(n/k)} \left[\log(1/\delta) + \log k + 1 \right]\right)$. We observe that the bound is $\propto \log(1/\delta)$ and $\propto \log(1/\epsilon) + 1/\epsilon$ in the DD case too.}

\textcolor{black}{The proposed PAC formulation for group testing explicitly characterizes the dependency of the bound on both $\epsilon$ and $\delta$ in a unified manner, unlike existing results. For example, we are able to provide insight into how the allowed approximation error scales with $k$ and $\epsilon$: for sufficiently small $\epsilon$, we get $g_\epsilon = d_\epsilon = \Theta(k)$ from the three results above.}

\textcolor{black}{The technical novelty of this work  is that it investigates the group testing problem from a new perspective, i.e., the PAC framework. Traditional PAC analysis methods require $\epsilon > 0$ and show that the sample complexity for PAC learnability varies as $1/\epsilon$~\cite{Mohri_Afshin_Talwalkar_2012}. In the context of group testing, we obtain the exact recovery bounds by setting $\epsilon = 0$. Therefore, the proof techniques used in classical PAC analysis can not be directly applied to group testing. We first bridge this gap between the traditional PAC framework and the analysis of group testing algorithms. We show an equivalence between the performance characterization in PAC learning and group testing in the exact recovery scenario, through Lemma~\ref{lem:defective_set_learnt_function_equivalence}. In the approximate recovery scenario, we relate the allowed number of false positive errors, $g_\epsilon$ and the false negative errors, $d_\epsilon$, with the approximation error tolerance probability, $\epsilon$, of the PAC framework (see~\eqref{eq:Prob_Error_FP_Relationship},~\eqref{eq:Prob_Error_FP_Relationship_sLength} and~\eqref{eq:Prob_Error_FN_Relationship}).}

\textcolor{black}{Existing results for the coupon collector problem~\cite{Mitzenmacher_Upfal_2005} cannot be directly applied to derive CBP bounds in the approximate recovery case. Therefore, we derive the expressions for the expected stopping time and the tail probability bound for a \emph{subset coupon collection problem} (SCCP), where one is interested in acquiring only a subset of coupons to complete the collection, in Lemma~\ref{lem:CCP_Lemma}. Further, we optimize the Chernoff parameter by bounding the error probability in two independent parts in the CBP analysis, thereby obtaining a tighter bound as compared to that in the literature for the exact recovery case~\cite{Chan_Jaggi_Saligrama_Agnihotri_2014}.}

\textcolor{black}{Lemma~\ref{lem:probability_CoMa} and Lemma~\ref{lem:DD_Interim_Results} (see part \ref{lem:DD_Interim_Results_c}) forms the basis for deriving the sufficiency bound for CoMa and DD algorithms: the exact computation of the probabilities that $g$ non-defective items are hidden and $d$ defectives items are unidentified (respectively) is characterized. Also, the order-wise analysis of DD bound involves simplification of an implicit equation in~\eqref{eq:DD_Bound} and solving the transcendental equation in~\eqref{eq:DD_transcendental}. Lastly, we give the sufficiency bounds for the three algorithms in terms of the allowed number of false positives or false negatives. In case the practitioner is unable to perform the required number of tests or they perform additional tests, the expressions in~\eqref{eq:CoMa_Bound},~\eqref{eq:CBP_s_Bound} and~\eqref{eq:DD_Bound} can be used to determine the increase or decrease in the false positives/negatives to be expected, respectively, and at different confidence levels.}

\section{PAC Learning View of Group Testing}
\label{sec:Function Learning View of Group Testing}
In this section, we cast the group testing problem as a function learning problem, also termed as learning from examples \cite{Valiant_1984}. Here, a target function $f \in \mathcal{C}$ is learnt using $m$ \emph{training examples} $(\mathbf{a}_i, f(\mathbf{a}_i)), i \in [m]$, with the inputs $\mathbf{a}_i$ drawn independently from a distribution $\mathcal{D}$. The training examples are fed to the learner, which then outputs an estimate of $f$, also called a hypothesis, and denoted by $h$. The error between $h$ and $f$ evaluated on unseen \emph{test} data is
\begin{equation}
	\label{eq:err_defn}
	e(h,f) = \mathbb{P}_{\mathbf{a} \sim \mathcal{D}}(h(\mathbf{a}) \neq f(\mathbf{a})).
\end{equation}
However, the quantity $e(h,f)$ is random because the $m$ training examples are drawn from $\mathcal{D}$. Therefore, we can ask how many training examples are sufficient to ensure that 
\begin{equation}
	\label{eq:PAC bound_standard}
	\mathbb{P}(e(h, f) > \epsilon) \leq \delta,
\end{equation}
where $\delta \in (0,1)$ and $1-\delta$ is called the \emph{confidence parameter}. Obviously, it is desirable to have small $\epsilon$ and $\delta$.

\vspace{-0.3cm}
\subsection{PAC Model for Group Testing}
\label{PAC Model for Group Testing}
Given a test matrix, $\mathbf{A}$, and the corresponding test outcomes, $y_i$, $i \in [m]$,  consider the problem of learning a hypothesis that can predict the group test outcomes with high confidence. More formally, we consider the $i$th row of the test matrix, $\mathbf{a}_i, i \in [m]$ as the input, and the outcome of the $i$th group test, $y_i$, as a \emph{label} associated with the $i$th training example: $(\mathbf{a}_i, y_i)$. Since the entries of $\mathbf{x}$ corresponding to the non-defective items are $0$, we can write $y_i$ as
\begin{align}
	y_i &= a_{ij_1}x_{j_1} \vee \ldots \vee a_{ij_k}x_{j_k} \nonumber \\
	\label{eq:function_learning_form_gt}
	    &= a_{ij_1} \vee a_{ij_2} \vee \ldots \vee a_{ij_k} \triangleq x(\mathbf{a}_i),
\end{align}
where $j_1, j_2, \ldots, j_k$ are the indices of $\mathbf{x}$ corresponding to the defective items. Thus, $x(\mathbf{a}_i)$ is a $k$-literal logical OR-ing function, and is our target function to learn. In computer science, this problem is referred to as the $k$-disjunctive function learning problem~\cite{Nick_1987}. \textcolor{black}{In this work, we focus on the analysis of group testing algorithms rather than developing decoding techniques inspired by the PAC framework. The target function space, denoted by $\mathcal{C}$, consists of all $k$-literal OR-ing functions, where $k$ literals are picked without replacement from $n$ literals in accordance with~\eqref{eq:function_learning_form_gt}.} 

The relationship between non-adaptive random pooling-based group testing and the function learning problem is summarized in Table~\ref{table:gt_vs_fl}, where $\mathcal{D} = \mathcal{B}(p)$ denotes a Bernoulli distribution with parameter $p \in (0,1)$. Another well-studied random pooling design is to uniformly and independently sample $s$ items with replacement in each group test~\cite{Chan_Jaggi_Saligrama_2011}.
\begin{table}[t]
	\centering
	\renewcommand{\arraystretch}{1.3}
	\caption{Group Testing as a PAC Learning Problem}
	\label{table:gt_vs_fl}
	\begin{tabular}{|c|c|c|c|}
		\hline
		& \textbf{Target} & \textbf{Training Examples} & \textbf{Hypothesis} \\
		\hline
		\textcolor{black}{\textbf{Group testing case}} & $x(\cdot)$ & $(\textbf{a}_{i} \overset{\text{i.i.d.}}{\sim} \mathcal{D}, y_i = x(\mathbf{a}_i))$ & $\hat{x}(\cdot)$ \\
		\hline
		\textcolor{black}{\textbf{PAC learning case}} & $f(\cdot)$ & $(\textbf{a}_i \overset{\text{i.i.d.}}{\sim} \mathcal{D}, f(\mathbf{a}_i))$ & $h(\cdot)$ \\
		\hline
	\end{tabular}
\end{table}

In the notation of group testing,~\eqref{eq:err_defn} can be written as
\begin{equation}
	\label{eq:err_defn_gt}
	e(\hat{x}(\cdot), x(\cdot)) = \mathbb{P}_{\mathbf{a} \sim \mathcal{D}}(\hat{x}(\mathbf{a}) \neq x(\mathbf{a})),
\end{equation}
which denotes the error probability on future group tests, i.e., after $m$ training samples are observed and the mapping $\hat{x}(.)$ is learnt. Thus, in the PAC learning view, we seek to determine the number of training examples, $m$, and a mapping from the training examples to a hypothesis, $\hat{x}(\cdot)$, such that with a confidence probability $1-\delta$, the error between $x(\cdot)$ and $\hat{x}(\cdot)$ is at most $\epsilon$ \cite{Valiant_1984}, \cite[Chapter~2]{Mohri_Afshin_Talwalkar_2012}, i.e.,
\begin{equation}
	\label{eq:PAC bound}
	\mathbb{P}(e(\hat{x}(\cdot), x(\cdot)) > \epsilon) \leq \delta
\end{equation}
\textcolor{black}{holds true, where $e(\cdot, \cdot)$, as defined in \eqref{eq:err_defn_gt}, is a random variable.}

\textcolor{black}{A fundamental difference between the PAC formulation of the group testing presented here as compared with the classical PAC-learning problem is that the data distribution, $\mathcal{D}$, is choosable. For example, the distribution can be set based on the hypothesis class, $\mathcal{C}$, i.e., based on the sparsity parameter, $k$. Even though the group testing \emph{algorithms} considered in this work do not use the knowledge of $k$ during the defective set recovery process, the design of the testing matrix and hence, the test data distribution, depends on this knowledge.}
 
We demonstrate that, when $\epsilon = 0$, the bounds on $m$ derived via the PAC model reduce to the exact recovery results derived in \cite{Chan_Jaggi_Saligrama_Agnihotri_2014, Aldridge_Balsassini_2014}. Note that, in classical group testing, the goal is to correctly identify the defective set, whereas in the PAC learning view of group testing, we seek to learn a hypothesis satisfying~\eqref{eq:PAC bound}. The following Lemma relates the PAC learning to group testing in the exact recovery case.\footnote{All the proofs are presented in Sec.~\ref{sec:Proofs of Theorems and Lemmas}.}
\begin{lemma}
	\label{lem:defective_set_learnt_function_equivalence}
	Let $\mathcal{D}$ be a distribution such that $\mathbb{P}_{\mathcal{D}}(a_{j} = 1) \in (0, 1),~j \in [n]$ and \textcolor{black}{$a_j$s are independent}. \textcolor{black}{Let $\mathcal{C}$ denote the set of all $k$-literal OR-ing functions in $n$-dimensional space, where $k < n$. Let $\hat{x}: \{0,1\}^n \to \{0,1\}$ (correspondingly $\hat{\mathcal{K}}$) be a function in $\mathcal{C}$ that is learnt using a set of $m$ training samples.} Then, $\hat{\mathcal{K}} = \mathcal{K}$ if and only if $\mathbb{P}_{\mathbf{a} \sim \mathcal{D}}(\hat{x}(\mathbf{a}) \neq x(\mathbf{a})) = 0$.
\end{lemma}
Lemma~\ref{lem:defective_set_learnt_function_equivalence} says that provided the marginal probability of every entry of the vector $\mathbf{a}$ is bounded in the open interval $(0,1)$ \textcolor{black}{and the entries of $\mathbf{a}$ are drawn independently}, the notions of recovery in~\eqref{eq:Error_Defective_Set_Estimate} and~\eqref{eq:PAC bound} when $\epsilon = 0$ are equivalent. Based on the above, we derive a sufficiency bound on $m$ for three group testing algorithms in the following sections.

\section{PAC Analysis with False Positive Errors}
\label{sec:PAC Analysis for False Positives}
In this section, we develop a PAC analysis for the case where only false positive errors occur, i.e., when $\mathcal{K} \subseteq \hat{\mathcal{K}}$. This is the case with two popular group testing algorithms: CoMa and CBP. The CoMa algorithm performs \emph{column-wise} decoding, while the CBP algorithm performs \emph{row-wise} decoding~\cite{Chan_Jaggi_Saligrama_Agnihotri_2014}. The two algorithms are mathematically equivalent in the sense that they always output the same $\hat{\mathcal{K}}$. However, the upper bound analysis in the two cases are different.

\vspace{-0.3cm}
\subsection{Bernoulli Test Design: The CoMa Algorithm}
\label{sec:The CoMa Algorithm}
The CoMa algorithm \cite{Chan_Jaggi_Saligrama_Agnihotri_2014} declares an item as defective if the ones in the column of the test matrix corresponding to that item \textcolor{black}{are a subset of} the ones in the outcome vector. Otherwise, the item is declared as non-defective. The algorithm never classifies a defective item as a non-defective. However, the estimate may contain false positives, which occur when non-defective items do not participate in any of the negative outcome tests in the training phase. Such items are also called \emph{hidden non-defectives}.

Suppose the hypothesis output by the CoMa algorithm has $G$ hidden non-defective items. Then, from~\eqref{eq:err_defn_gt}, the probability that $\hat{x}(\cdot) \equiv \hat{x}$ differs from $x(\cdot) \equiv x$ for the next group test, defined as $\mathbb{P}_{\mathbf{a}_i \sim \mathcal{B}(p)} (\hat{x}(\mathbf{a}_i) \neq x(\mathbf{a}_i))$, is a function of $G$. Suppose we want the error between $\hat{x}$ and $x$ to be at most $\epsilon$. In turn, this requires $G \leq g_\epsilon$, where $g_\epsilon$ can be computed from
\begin{align}
	{\mathbb{P}}_{\textbf{a}_i\sim \mathcal{B}(p)}\left(\hat{x}(\textbf{a}_i ) \neq x(\textbf{a}_i)\right) &= (1-(1-p)^{G})(1-p)^k \leq \epsilon \nonumber \\
	\Rightarrow g_\epsilon &= \Biggl\lfloor \frac{\log\left(1-\epsilon/\left(1-p\right)^k\right)}{\log(1-p)} \Biggr\rfloor.	\label{eq:Prob_Error_FP_Relationship}
\end{align}
Thus, the bound in \eqref{eq:PAC bound} reduces to $\mathbb{P}(G > g_\epsilon) \leq \delta$. In order to proceed further, we need the following Lemma.
\begin{lemma}
	\label{lem:probability_CoMa}
	The probability that a fixed set of $g~(1 \leq g < k)$ non-defective items remain hidden in all $m$ tests is given~by
	\begin{equation}
		\label{eq:probability_CoMa_equation}
		\mathbb{P}^h_{g}(m) = \left(1 - (1-p)^k + (1-p)^{g+k}\right)^m.
	\end{equation}
\end{lemma}
From the above Lemma, we obtain the following sufficiency condition on $m$ for CoMa under the PAC model:
\begin{theorem}\label{thm:CoMa_Bound}
	The sufficient number of tests such that the predicted outcome based on the estimated defective set does not agree with the true outcome on future group tests with probability at most $\epsilon$ and confidence parameter $1-\delta$~is
	\begin{equation}
		\label{eq:CoMa_Bound}
		m_S =\frac{\log{n-k \choose g_\epsilon + 1}+\log\frac{1}{\delta}}{\log\left(1/(1-(1-p)^k+(1-p)^{g_\epsilon + k + 1})\right)},
	\end{equation}
	with $g_\epsilon$ as given by \eqref{eq:Prob_Error_FP_Relationship}.
\end{theorem}

\subsubsection{Optimum Bernoulli Parameter}
\label{sec:Optimum_p}
\textcolor{black}{We can determine the optimum value of $p$ for which~\eqref{eq:CoMa_Bound} is minimized by solving the following mixed-integer non-linear program (MINLP):}
	\begin{align}
		\hat{m}_S, \hat{p}_\text{opt}, \hat{g}_\epsilon &= \underset{\underset{g_\epsilon \in \mathbb{Z}_+ \cup \{0\}}{m \in \mathbb{Z}_+,~p \in (0,1)}}{\argmin}~m \nonumber \\
		&~~~~~~\text{s.t.}~(1-(1-p)^{g_\epsilon})(1-p)^k \leq \epsilon \nonumber \\&~~~\text{and} \nonumber \\
		\label{eq:CoMa_MINLP}
		&{n-k \choose g_\epsilon + 1}\left(1\!-\!(1\!-\!p)^k\!+\!(1\!-\!p)^{g_\epsilon\!+\!1\!+\!k}\right)^m \leq \delta.
\end{align}
\textcolor{black}{The constraints above correspond to the error probability ($\epsilon$) and confidence ($1-\delta$) requirements. Since the problem does not admit a closed-form solution, we solve~\eqref{eq:CoMa_MINLP} using a grid-search over $m$, $p$ and $g_\epsilon$, ensuring that we are varying them over the feasible range.}

\begin{figure*}[t]
	\centering
	\includegraphics[width=1.0\linewidth]{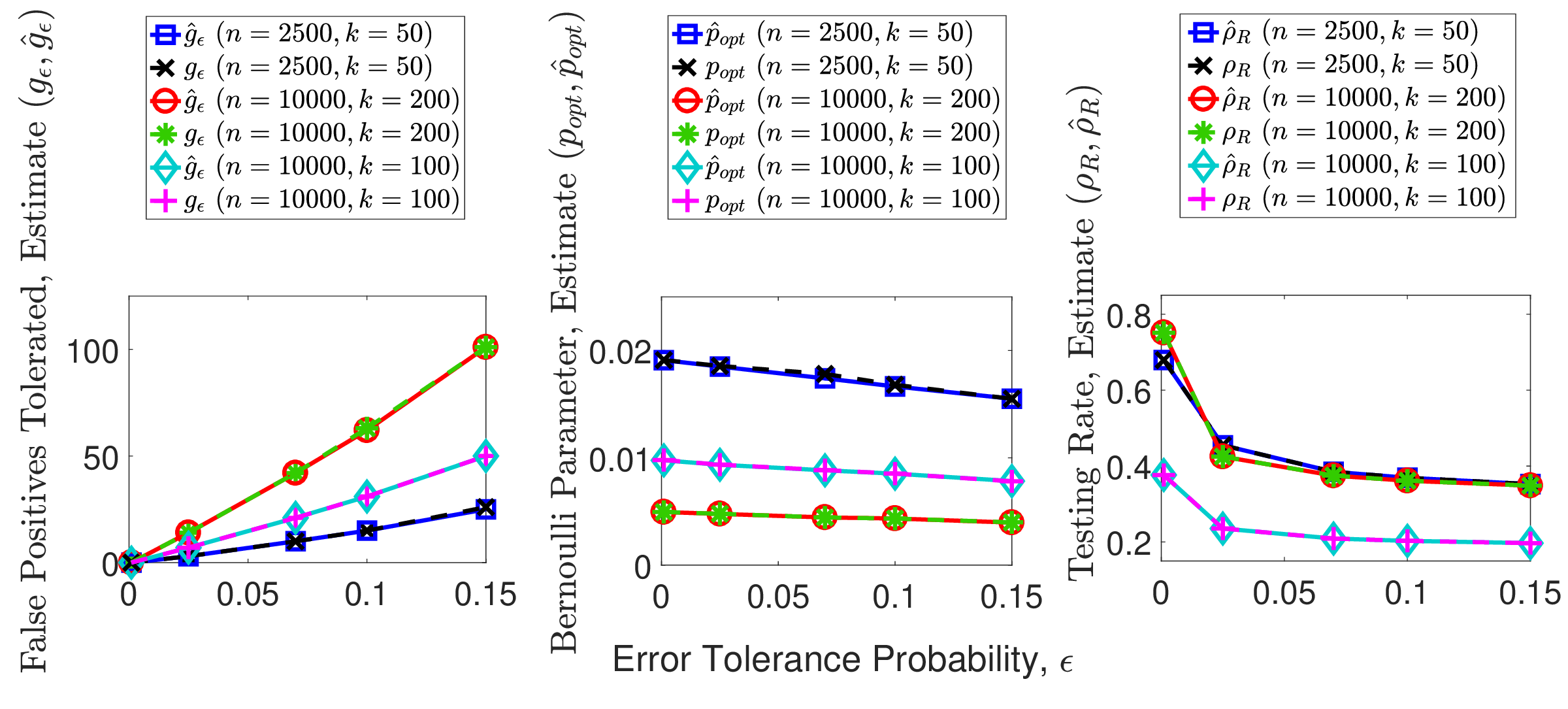}
	\vspace{-0.3in}
	\caption{Comparison of the solution of MINLP using grid-search vs. the implicit equations for $m_S$, $p_\text{opt}$ and $g_\epsilon$ at $\delta = 0.01$ when $(n, k) \in \{(2500, 50), (10000, 200), (10000, 100)\}$ over various values of the approximation error tolerance, $\epsilon$.}
	\label{fig:coma_minlp}
\end{figure*}

\textcolor{black}{In order to obtain insights, we also minimize~\eqref{eq:CoMa_Bound} using a heuristic fixed-point iteration method. Note that~\eqref{eq:CoMa_Bound} holds for any $g_\epsilon \in \mathbb{Z}_+$. Further, when $g_\epsilon=0$ (i.e., for exact recovery), $p_{\text{opt}}$ in~\eqref{eq:CoMa_pOpt} reduces to $1/(k+1)$ or $O(1/k)$ for large $k$. We can use $1/(k+1)$ as an initial value for $p$ and find the largest number of false positive errors $g_\epsilon$ for which the $\epsilon$-constraint is satisfied from \eqref{eq:Prob_Error_FP_Relationship}. Then, for the given $g_\epsilon$, we seek the $p$ for which the term $\left(1-(1-p)^k+(1-p)^{g_\epsilon + 1 + k}\right)$ is minimized, since this results in the smallest $m$ for which the $\delta$-constraint is satisfied. Differentiating this term and setting it equal to zero yields the optimum value of $p$ as}
\begin{equation}
	\label{eq:CoMa_pOpt}
	p_{\text{opt}}\!=\!1\!-\!\sqrt[{g_\epsilon\!+\!1}]{k/(k\!+\!g_\epsilon\!+\!1)}.	
\end{equation}
\textcolor{black}{We now iterate between \eqref{eq:CoMa_pOpt} and \eqref{eq:Prob_Error_FP_Relationship} to obtain the optimal $g_\epsilon$ and $p_\text{opt}$. Although it appears hard to prove analytically, we find that this procedure converges, and yields the globally optimal solution. We compare the numerical solution, $\hat{g}_\epsilon$ and $\hat{p}_\text{opt}$, obtained by solving \eqref{eq:CoMa_MINLP} with that obtained using the fixed-point procedure, namely, $g_\epsilon$ and $p_\text{opt}$. Also, we compare the value of $m_S$ computed from Theorem~\ref{thm:CoMa_Bound} along with~\eqref{eq:CoMa_pOpt} and~\eqref{eq:Prob_Error_FP_Relationship} with that obtained by solving~\eqref{eq:CoMa_MINLP}, i.e., $\hat{m}_S$.} The computed values and the solution of the MINLP are plotted over various values of the approximation error tolerance, $\epsilon$, with $\delta = 0.01$ for a collection $(n, k) \in \{(2500, 50), (10000, 200), (10000, 100)\}$ in Fig.~\ref{fig:coma_minlp}. It can be observed that the expression for $p_\text{opt}$ as given in~\eqref{eq:CoMa_pOpt} is consistent with the expression for $g_\epsilon$ in~\eqref{eq:Prob_Error_FP_Relationship} which further yields the optimum $m_S$ as given by~\eqref{eq:CoMa_Bound}.

\subsubsection{Order-Wise Analysis of~\eqref{eq:CoMa_Bound}}
\label{sec:CoMa_Order-wise_analysis}
We note that the $\log{n-k \choose g_\epsilon+1}$ term appearing in the expression for $m_S$ is similar to the $\log{n-k \choose \tau}$ term obtained in~\cite{Scarlett_Cevher_2017, Aldridge_GT_IT_2019}, where $\tau$ denoted the number of false positives. Further, the achievability bound in~\cite{Scarlett_Cevher_2017} is characterized by the conditional mutual information term in the denominator. \textcolor{black}{In this sub-section, we elucidate the behavior of the denominator in~\eqref{eq:CoMa_Bound} for large $n$ and $k$, resulting in an order-wise characterization of the achievability bound.}

First, we can relate the CoMa bound for the exact recovery case in \cite{Chan_Jaggi_Saligrama_Agnihotri_2014} to our PAC bound. From~\eqref{eq:Prob_Error_FP_Relationship}, we see that $\epsilon = 0$ implies $g_\epsilon = 0$. In the \emph{sub-linear} regime, $k = \Theta(n^\beta), \beta \in (0,1)$, the numerator in \eqref{eq:CoMa_Bound} can be upper bounded by $\log n + \log(1/\delta)$. Let $z \triangleq 1-(1-p)^k + (1-p)^{k+1}$. \textcolor{black}{Using $\log(1/\tilde{z}) \geq 1-\tilde{z},~\tilde{z} > 0$}, the denominator in \eqref{eq:CoMa_Bound} can be lower bounded as follows:
\begin{equation}
	\log\left(1/\tilde{z}\right) \geq p(1-p)^k = \frac{1}{ek} - O\left( \frac{1}{ek^2} \right)
	\geq \frac{1}{e(k+1)} \label{eq:Denominator_CoMA_Bound_Reduction}
\end{equation}
where $p=1/k$ \cite{Chan_Jaggi_Saligrama_Agnihotri_2014}, and Laurent's series is used in the penultimate step as $k$ gets large. Combining the reduced expressions, \textcolor{black}{we obtain that $m = e(k + 1)\left(\log n + \log(1/\delta)\right)$ are sufficient for exact recovery.} From the sufficiency result for $m$ in \cite[Theorem 4]{Chan_Jaggi_Saligrama_Agnihotri_2014}, with $\delta^\prime\!\triangleq\!n^{-\delta}$, $m = ek(\log n + \log(1/\delta^\prime))$ is sufficient, which is similar to our bound at $\epsilon\!=\!0$.

\textcolor{black}{Second, we analyze how the achievability bound in~\eqref{eq:CoMa_Bound} behaves as as a function of $\delta$ and $\epsilon$ when $n$ and $k$ grow large. Define $z \triangleq (1-p)^k - (1-p)^{g_\epsilon+k+1}$. Using ${n-k \choose g_\epsilon+1} \leq (e(n-k)/(g_\epsilon+1))^{g_\epsilon+1}$ along with the fact that $\log(1/(1-z)) \geq z$, $z \in [0, 1]$ in~\eqref{eq:CoMa_Bound}, we see that
\begin{equation}
	\label{eq:CoMa_Bound_geps_numreduce}
	m_S = \frac{(g_\epsilon+1)\log\left(\frac{n-k}{g_\epsilon+1}\right) + (g_\epsilon+1) + \log\left(\frac{1}{\delta}\right)}{(1-p)^k\left(1-(1-p)^{g_\epsilon+1}\right)},
\end{equation}
tests are sufficient to ensure no more than $g_\epsilon$ errors with confidence $1-\delta$. Set $p = 1/k$ in~\eqref{eq:CoMa_Bound_geps_numreduce}. Using $(1-x)^r \leq e^{-xr}$ for $x \in [0,1], r \geq 0$ and $1-e^{-y} \geq y/2$, for $y \in [0,1]$, we lower bound the second factor in the denominator of~\eqref{eq:CoMa_Bound_geps_numreduce} with $x = p$, $r = g_\epsilon+1$ and $y = (g_\epsilon+1)/k$ as:
\begin{align}
	\label{eq:CoMa_geps_reduce_den_approx}		
	1-(1-p)^{g_\epsilon+1} \geq \frac{g_\epsilon+1}{2k}.
\end{align}
Using $(1-p)^k \to 1/e$ for large $k$ and $\log(n-k) < \log n$, along with~\eqref{eq:CoMa_geps_reduce_den_approx} in~\eqref{eq:CoMa_Bound_geps_numreduce}, we get
\begin{align}
	m_S &= 2ke\left[\log\left(\frac{n}{g_\epsilon+1}\right) + 1 + \frac{\log\left(\frac{1}{\delta}\right)}{g_\epsilon+1}\right] \nonumber \\
	\label{eq:CoMa_approx_asymptotic}
	&= 2ke\left[\log\left(\frac{n}{ke\epsilon+1}\right) + 1 + \frac{\log\left(\frac{1}{\delta}\right)}{ke\epsilon+1}\right],
\end{align}
where, we have used $p = 1/k$ in~\eqref{eq:Prob_Error_FP_Relationship} and  $\log(1-x) \to -x$, $x \to 0$ along with $(1-p)^k \to 1/e$ as $n$ and $k$ grow, to get $g_\epsilon = ke\epsilon$. From~\eqref{eq:CoMa_approx_asymptotic}, we see that $m_S \propto \log(1/\delta)/\epsilon$ for very small $\delta$. Further, we see that the dependency of our bound on $\epsilon > 0$ is $\propto\!\left(\log(1/\epsilon) + 1/\epsilon\right)$. On the other hand, we get $m_S \approx 2ke\left(\log n + 1 + \log(1/\delta)\right)$ by setting $\epsilon = 0$, i.e., for the exact recovery case, accounting for the confidence parameter~$\delta$.}

\subsection{Near-Constant Row-Weight Design: The CBP Algorithm}
\label{sec:The CBP Algorithm}
The CBP algorithm~\cite{Chan_Jaggi_Saligrama_Agnihotri_2014} declares all the items participating in the group tests with negative outcomes as non-defective, and the remaining items as defective. It is clear that the CBP algorithm, like CoMa, makes only false positive errors. Moreover, both CoMa and CBP algorithms are equivalent in the sense that they lead to the same set of items being declared positive. In the case of the CBP algorithm, we consider a near-constant row-weight design for the random testing matrix, and, consequently, the analysis technique is different from the CoMA case. As we shall see from the numerical results, the CoMa bound is tighter than the CBP bound in the exact recovery case. On the other hand, the CBP bound is tighter in the approximate recovery case, i.e., as $g_\epsilon$ increases. Therefore, it is of interest to study CBP bounds separately. We also note that the CoMa analysis yields an analytical expression for the optimum parameters of the test design, as seen earlier.

The CBP algorithm is related to the Coupon Collector Problem (CCP), where the goal is to collect distinct \emph{coupons} to obtain a set of all available coupons~\cite{Chan_Jaggi_Saligrama_Agnihotri_2014}. More precisely, there is one of $n$ distinct types of coupon, say, inside each cereal box. How many cereal boxes should a person purchase in order to collect all the $n$ coupons?

The expected stopping time, i.e., the average number of purchases required to succeed (with replacements, as coupons can repeat across purchases) is $nH_n,$ where $H_n \triangleq \sum_{i=1}^{n}~1/i$ denotes the $n$th Harmonic number for any $n \in \mathbb{N}$ and $H_0 \triangleq 0$. A well known asymptotic approximation for $H_n$ is $H_n \approx \log n + \gamma + 1/2n + O(1/n^2) \approx O(\log n)$, where $\gamma \approx 0.5772$ is the Euler–Mascheroni constant. Therefore, the expected stopping time is $O(n \log n)$ for sufficiently large $n$. \textcolor{black}{One can bound $H_n$ as $\log n + \gamma < H_n < \log(n+1) + \gamma$.} Also, the probability that the stopping time exceeds $\chi n \log n$ is at most $n^{-\chi+1}$, for~$\chi > 1$~\cite{Motwani_Raghavan_1995, Mitzenmacher_Upfal_2005}.

Before we present the main result, on similar lines to \cite{Chan_Jaggi_Saligrama_Agnihotri_2014}, we relate the CBP algorithm to the CCP as follows. The CBP algorithm collects items from a sequence of tests. Consider an $s$-length test vector whose entries index the items being pooled in a given test. The $s$ items are chosen with replacement.\footnote{A given item can be potentially picked more than once apart from being picked once or not picked at all. By assigning the number of times an item is picked in the $i$th test to the $(i,j)$th entry of the test design matrix, the authors in~\cite{Chang_Chen_Guo_Huang_2015} present a multi-group testing model using standard (not Boolean) arithmetic. However, in our work, the $(i,j)$th entry of the testing matrix is set to one irrespective of whether the $j$th item is picked once or more than once in $i$th test, and is set to zero otherwise.}

Following~\cite{Chan_Jaggi_Saligrama_Agnihotri_2014}, it is clear that there is a natural bijection between the $s$-length vector and $n$-length row vector $\textbf{a}_i$ of the $i$th group test. Since the probability of an item occurring at any location of the $s$-length vector is uniform and independent, and this property holds across the tests, the items in any sub-sequence of $m^\prime$ tests with negative outcomes may be viewed as the outcome of a process of selecting a single chain of $sm^\prime$ coupons. Since the outcome of a single test with $s$ items is negative with probability $((n-k)/n)^s$, in order for the CBP algorithm to succeed with $m$ tests, in the exact recovery case, we require
\begin{equation}
	\label{eq:CBP_CCP_Expected_Inequality}
	ms\left(\frac{n\!-\!k}{n}\right)^s \geq (n\!-\!k)H_{n\!-\!k} \geq (n\!-\!k)\left[ \log(n\!-\!k)\!+\!\gamma \right].
\end{equation}

For approximate recovery, i.e., $g_\epsilon$ errors, it suffices to \emph{stop collecting} items participating in negative outcome tests once we collect $n-k-g_\epsilon$ non-defective items.  Lemma~\ref{lem:CCP_Lemma} presents the expected stopping time and tail probability in this case.
\begin{lemma}\label{lem:CCP_Lemma}
	For the coupon collector problem with $w$ distinct coupons, with each coupon being picked in an equally likely and independent fashion, if any subset containing $w-g$ distinct coupons are sufficient to complete the collection, then\\
	(a) The expected stopping time is $w[\log w + \gamma - H_g]$, where $\gamma$ is the Euler-Mascheroni constant and $H_g$ is the $g$th Harmonic number as defined earlier.\\
	(b) For any $\chi > 1$, the stopping time exceeds 
	$\chi w[\log w + \gamma - H_g]$ with probability at most $w^{(g+1)(-\chi+1)} \frac{e^{(g+1)\chi [H_g - \gamma] + g}}{(g+1)^{(g+1)}}$.
\end{lemma}

From Lemma~\ref{lem:CCP_Lemma} (a), the RHS of~\eqref{eq:CBP_CCP_Expected_Inequality} can be modified with $g=g_\epsilon$ and $w=n-k$ to obtain
\begin{equation}
	\label{eq:CBP_CCP_Approximate_Recovery_Inequality}
	ms\left(\frac{n-k}{n}\right)^s \geq (n-k)\left[\log(n-k) + \gamma - H_{g_\epsilon}\right].
\end{equation}

With this background, we are ready to present the result on the sufficient number of tests required by the CBP algorithm in the PAC framework. Since CBP can only make false positive errors, the PAC equation \eqref{eq:PAC bound} takes the form $\mathbb{P}(G > g_\epsilon) \leq \delta$, similar to the CoMa algorithm. 

We note that an error in the prediction occurs when none of the $k$ defectives out of $n$ items participate in the group test and at least one of $G$ hidden non-defectives participates, making the predicted outcome $1$ whereas the true outcome is $0$. On similar lines as~\eqref{eq:Prob_Error_FP_Relationship}, the largest value of $g_\epsilon$ such that the error between $\hat{x}$ and $x$ is at most $\epsilon$ can be computed as
\begin{equation}
	\mathbb{P}_{\textbf{a}_i \sim \mathcal{S}}(\hat{x}(\textbf{a}_i) \neq x(\textbf{a}_i)) = \left( 1\!-\!\frac{k}{n} \right)^s \left[ 1\!-\!\left( 1\!-\!\frac{G}{n\!-\!k} \right)^s \right] \leq \epsilon, \nonumber
\end{equation} 
where $\mathcal{S}$ denotes a uniform distribution over all $s$-length vectors. This yields
\begin{equation}
	\label{eq:Prob_Error_FP_Relationship_sLength}
	g_\epsilon = \left\lfloor (n-k) \left[ 1 - \left( 1 - \frac{\epsilon}{\left(1-\frac{k}{n}\right)^s} \right)^{1/s} \right] \right\rfloor.
\end{equation} 
The following theorem characterizes the sufficient number of tests required by the CBP algorithm in the PAC framework.
\begin{theorem}\label{thm:CBP_s_Bound}
	Suppose $s$ items are chosen with replacement in each group test. The sufficient number of tests required by the CBP algorithm such that the predicted outcome based on the estimated defective set does not agree with the true outcome on future group tests with probability at most $\epsilon$ and confidence parameter $1-\delta$~is
	\begin{equation}
		\label{eq:CBP_s_Bound}
		m_S = \frac{\chi(n-k)}{(1-\eta)s\left(\frac{n-k}{n}\right)^s} \left[\log(n-k)+\gamma-H_{g_\epsilon}\right],
	\end{equation}
	where $g_\epsilon$ is given by~\eqref{eq:Prob_Error_FP_Relationship_sLength}, 
	\begin{equation}
		\label{eq:CBP_s_Bound_Chi}
		\chi = \frac{\left[\frac{\log\left(\frac{1}{c\delta}\right)}{g_\epsilon+1} + \frac{g_\epsilon}{g_\epsilon+1} + \log\left( \frac{n-k}{g_\epsilon+1} \right)\right]}{\log(n-k)+\gamma-H_{g_\epsilon}},
	\end{equation}
$\eta = (-C + \sqrt{C^2 +4C})/2 \in (0,1)$, with
	\begin{equation}
		\label{eq:CBP_CCP_C_Definition}
		C \triangleq \frac{\log\left(\frac{1}{(1-c)\delta}\right)}{\left(\frac{n\!-\!k}{s}\right) \left[ \frac{\log\left(\frac{1}{c\delta}\right)}{g_\epsilon\!+\!1}\!+\!\frac{g_\epsilon}{g_\epsilon\!+\!1}\!+\!\log\left( \frac{n\!-\!k}{g_\epsilon\!+\!1} \right) \right]},
	\end{equation}
	and $c \in (0,1)$ is a design parameter.
\end{theorem}

\subsubsection{Order-Wise Analysis - Exact Recovery Case}
\label{sec:CBP_Order-wise_analysis_exact}
\textcolor{black}{We now discuss the order-wise behavior of~\eqref{eq:CBP_s_Bound} for large $n$ and $k$, and specialize the result to the case when $\epsilon = 0$.} Note that differentiating~\eqref{eq:CBP_CCP_Expected_Inequality} with respect to $s$ and setting the derivative equal to $0$ yields $s^* = 1/\log(n/(n-k))$.
Using this value of $s^*$, we obtain the following corollary.
\begin{cor}\label{cor:CBP_s_Bound_sStar}
	With $\chi$ and $\eta$ as specified in Theorem~\ref{thm:CBP_s_Bound} computed at $s = s^* \triangleq 1/\log(n/(n-k))$, the sufficient number of tests required by the CBP algorithm such that the predicted outcome based on the estimated defective set does not agree with the true outcome on future group tests with probability at most $\epsilon$ and confidence parameter $1-\delta$~is
	\begin{equation}
		\label{eq:CBP_s_Bound_Cor_sStar}
		m_S = \frac{\chi k}{(1-\eta)}\left(\frac{n}{n-k}\right)^{s^*} \left[\log(n-k)+\gamma-H_{g_\epsilon}\right],
	\end{equation} 
	with $g_\epsilon$ as given by~\eqref{eq:Prob_Error_FP_Relationship_sLength}.
\end{cor}

Recall that $g_\epsilon = 0$ corresponds to the exact recovery case as considered in the  literature~\cite{Chan_Jaggi_Saligrama_2011}. Also, for reasonably large $n$, $(n/(n-k))^{s^*} \to e$ \textcolor{black}{from below}. \textcolor{black}{Substituting for $\chi$ from~\eqref{eq:CBP_s_Bound_Chi} with $g_\epsilon = 0$ into~\eqref{eq:CBP_s_Bound_Cor_sStar} leads to
\begin{equation}
	\label{eq:CBP_s_Bound_Cor_sStar_Exact}
	m_S = \frac{ek}{(1-\eta)} \left[\log(n-k)+\log\left(\frac{2}{\delta}\right)\right].
\end{equation}
We now set $s = s^*$, $c = 1/2$, and $g_\epsilon = 0$ and reduce~\eqref{eq:CBP_CCP_C_Definition} to
\begin{align}
	\label{eq:CBP_C_Exact}
	C &= \frac{s^*\log\left( \frac{2}{\delta} \right)}{(n-k)\left[ \log\left(\frac{2}{\delta}\right) + \log(n-k) \right]} \\
	\label{eq:CBP_C_Approx_2_Regimes}
      &\approx
    \begin{cases}
	  		~~~~~\frac{s^*}{(n-k)}~&,~\delta \ll \frac{2}{n-k} \\
	  		\frac{s^* \log\left(\frac{2}{\delta}\right)}{(n-k)\log(n-k)}~&,~\delta \gg \frac{2}{n-k}\\
	\end{cases}
\end{align}
Using the fact that $s^*/(n-k) \ll 1$ holds for large $n$ and in the sub-linear regime along with~\eqref{eq:CBP_C_Exact} and~\eqref{eq:CBP_C_Approx_2_Regimes}, a Taylor series approximation yields
\begin{subnumcases}{\eta\!\approx\!}
	\sqrt{\frac{s^*}{(n\!-\!k)}}&,~$\delta\!\ll\!\frac{2}{n\!-\!k}$ \label{eq:CBP_rho_sStar_Dominant_n}
	\\
	\sqrt{\frac{s^* \log\left(\frac{2}{\delta}\right)}{(n\!-\!k)\log(n\!-\!k)}}&,~$\delta\!\gg\!\frac{2}{n\!-\!k}$ \label{eq:CBP_rho_sStar_Dominant_delta}
	\\
	\sqrt{\frac{s^* \log\left(\frac{2}{\delta}\right)}{(n\!-\!k)\left[\log\left(\frac{2}{\delta}\right)\!+\!\log(n\!-\!k)\right]}}&,~$\delta\!\sim\!\frac{2}{n\!-\!k}$ \label{eq:CBP_rho_sStar_SimMag}
\end{subnumcases} 
where, in~\eqref{eq:CBP_rho_sStar_SimMag}, $a \sim b$ is used to signify that $a$ and $b$ are of the same order.}

From~\eqref{eq:CBP_rho_sStar_SimMag}, we see that $\eta$ scales as $O\left(1/\sqrt{n\log n}\right)$ for large $n$.
From the sufficiency result for $m$ in \cite[Theorem 3]{Chan_Jaggi_Saligrama_Agnihotri_2014}, with $\delta^\prime \triangleq 2n^{-\delta}$, $m = 2ek(\log n + \log(2/\delta^\prime))$ is sufficient, which is similar to our PAC bound \eqref{eq:CBP_s_Bound_Cor_sStar_Exact} when $\epsilon=0$ (\textcolor{black}{also see the discussion below~\eqref{eq:CBP_final_approx}}). Explicitly computing the optimum $\eta$ as in~\eqref{eq:CBP_rho_sStar_Dominant_n},~\eqref{eq:CBP_rho_sStar_Dominant_delta} and~\eqref{eq:CBP_rho_sStar_SimMag} instead of using a nominal value $\eta = 1/2$~\cite{Chan_Jaggi_Saligrama_Agnihotri_2014} yields approximately a factor of $2$ improvement in the testing rate when $g_\epsilon = 0$.

\subsubsection{Order-Wise Analysis - Approximate Recovery Case}
\label{sec:CBP_Order-wise_analysis_approximate}
We first discuss the behavior of $\eta$ when $g_\epsilon > 0$, for large $n$ and $k$. When $g_\epsilon > 0$, we have $1/2 \leq g_\epsilon/(g_\epsilon+1) < 1$ and further, dropping $g_\epsilon/(g_\epsilon+1)$ in the denominator of~\eqref{eq:CBP_CCP_C_Definition} only increases the value of $C$ and hence $\eta$. Therefore, dropping the $g_\epsilon/(g_\epsilon+1)$ term in the expression for $C$ does not violate the sufficiency of the bound in~\eqref{eq:CBP_s_Bound_Cor_sStar}. \textcolor{black}{As before, we set $c = 1/2$ and $s = s^*$ and observe that $s^*/(n-k) \ll 1$ holds for large $n$ and therefore, using a Taylor series approximation in the sub-linear regime, the Chernoff parameter, $\eta$, can be approximated as
\begin{equation}
	\label{eq:CBP_rho_sStar_Approx_Recovery}
	\eta \approx \sqrt{\frac{s^*\log\left(\frac{2}{\delta}\right)}{(n\!-\!k)\left[\frac{\log\left(\frac{2}{\delta}\right)}{g_\epsilon\!+\!1}\!+\!\log\left(\frac{n\!-\!k}{g_\epsilon\!+\!1}\right)\right]}}.
\end{equation}}

\textcolor{black}{From~\eqref{eq:CBP_rho_sStar_Approx_Recovery}, we observe that $\eta$ scales as $O(1/\sqrt{n\log n})$ for any fixed $\epsilon \geq 0$, similar to its behavior in the exact recovery scenario. Further,~\eqref{eq:CBP_rho_sStar_Approx_Recovery} reduces to~\eqref{eq:CBP_rho_sStar_SimMag} at $\epsilon = 0$.}

\textcolor{black}{We now characterize the order-wise behavior of our bound. We start by substituting~\eqref{eq:CBP_s_Bound_Chi} in~\eqref{eq:CBP_s_Bound_Cor_sStar}, to get
\begin{equation}
	\label{eq:CBP_first_asymptotic_approx}
	m_S = \frac{ek}{(1-\eta)} \left[\frac{\log\left(\frac{1}{c\delta}\right)}{g_\epsilon+1} + \frac{g_\epsilon}{g_\epsilon+1} + \log\left( \frac{n-k}{g_\epsilon+1} \right)\right], 
\end{equation}
where $(n/(n-k))^{s^*} \to e$ from below as $n$ grows large. Using $\eta = (-C+\sqrt{C^2+4C})/2$ from Theorem~\ref{thm:CBP_s_Bound}, we get
\begin{equation}
	\label{eq:CBP_CCP_eta_ratio_bound}
	\frac{1}{1-\eta} \leq 2+C.
\end{equation}
Using $\log(1/x) \to 1-(1/x)$ for $x \to 1$ with $x = (n-k)/n$, we get $s^* \to (n-k)/n$, for large $n$. Therefore, we use
\begin{equation}
	\label{eq:CBP_CCP_C_Aysmptotic_Bound}
	C = \frac{\log\left(\frac{1}{(1-c)\delta}\right)}{k\left[\frac{\log\left(\frac{1}{c\delta}\right)}{g_\epsilon+1} + \frac{g_\epsilon}{g_\epsilon+1} + \log\left( \frac{n-k}{g_\epsilon+1} \right)\right]},
\end{equation}
along with~\eqref{eq:CBP_CCP_eta_ratio_bound} in~\eqref{eq:CBP_first_asymptotic_approx} to get
\begin{align}
	m_S &= 2ke\left[\log\left(\frac{n-k}{g_\epsilon+1}\right) + \frac{g_\epsilon}{g_\epsilon+1} + \frac{\log\left(\frac{1}{c\delta}\right)}{g_\epsilon+1}\right] \nonumber\\
	\label{eq:CBP_second_asymptotic_approx}
	&~~~~~~~~+ e\log\left(\frac{1}{(1-c)\delta}\right).
\end{align}
Observe that $g_\epsilon/(g_\epsilon+1) < 1$ and $\log(n-k) < \log n$, when $n > k > 0$. Setting $c = 1/2$ in~\eqref{eq:CBP_second_asymptotic_approx}, we get
\begin{equation}
	\label{eq:CBP_third_asymptotic_approx}
	m_S = 2ke\left[\log\left(\frac{n}{g_\epsilon\!+\!1}\right) \!+\!1\!+\!\log\left(\frac{2}{\delta}\right)\!\left[\frac{1}{g_\epsilon\!+\!1}\!+\!\frac{e}{2k}\right]\right].
\end{equation}
As mentioned earlier, $(n/(n-k))^{s^*} \to e$ for 
large $n$. Also, $(1-x)^{(1/x)-1} \to 1/e$ as $x \to 0$ with $x = k/n$ and $(1-x)^n \to 1-nx$ for $x < 1$ and $|nx| \ll 1$. Therefore, $g_\epsilon = ke\epsilon$ from~\eqref{eq:Prob_Error_FP_Relationship_sLength}. Substituting for $g_\epsilon$ in~\eqref{eq:CBP_third_asymptotic_approx}, we arrive at
\begin{align}
	\label{eq:CBP_final_approx}
	m_S = 2ke\left[\log\left(\frac{n}{ke\epsilon\!+\!1}\right) \!+\!1\!+\!\log\left(\frac{2}{\delta}\right)\!\left[\frac{1}{ke\epsilon\!+\!1}\!+\!\frac{e}{2k}\right]\right].
\end{align}
For a given $n$, $k$, and $\delta$, we see that the dependency of our CBP bound on $\epsilon > 0$ is $\propto\!\left(\log(1/\epsilon) + 1/\epsilon\right)$, similar to the CoMa bound.  On the other hand, we get $m_S = 2ke[\log n + 1 + (1+e/2k)\log(2/\delta)]$ for the exact recovery case by setting $\epsilon = 0$, showing that the analysis in this subsection is inclusive of the exact recovery case.}

\subsubsection{Utility of Optimizing the Chernoff Parameter}
\label{sec:Optimum_eta}
We now discuss the effectiveness of the approximation to the Chernoff parameter, $\eta$, for both $g_\epsilon = 0$ and $g_\epsilon > 0$. To this end, we start by defining the \emph{testing rate} as $\rho_R \triangleq m_S/n$. Note that $\rho_R$ denotes the sufficient number of tests per item.

Fig.~\ref{fig:cbp_exact_bounds_approx} shows the testing rates, $\rho_R$ over various values of the \emph{log confidence} parameter, $\log(1/\delta)$ for both exact recovery and approximate recovery cases. For illustration, we choose $g_\epsilon = 5$ for our discussion on the approximate recovery case in this subsection. The plots are generated with $n = 2500$ and $k = 50$, corresponding to an error tolerance of $10\%$ in the approximate recovery case.

\begin{figure}[t]
	\centering
	\includegraphics[width=0.8\linewidth]{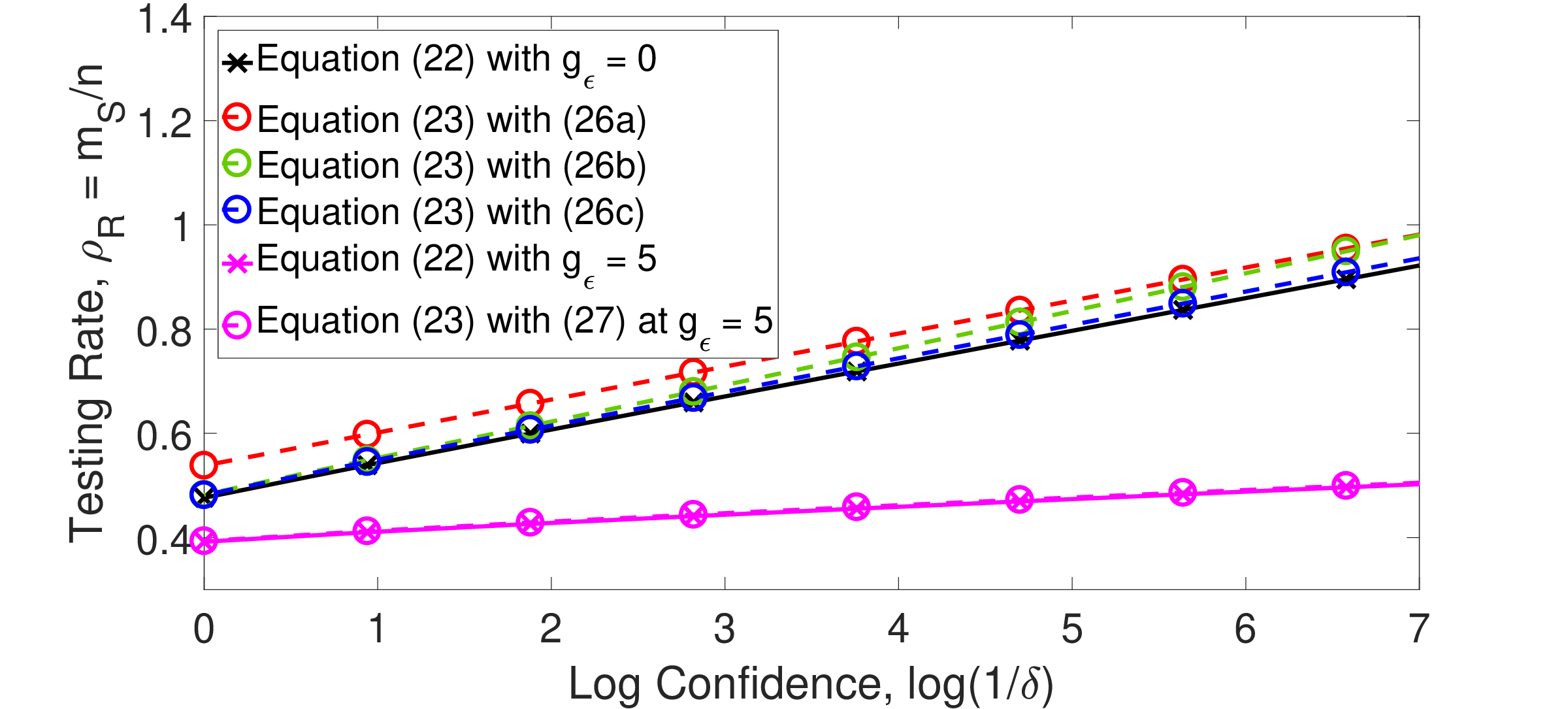}
	\caption{Comparison of the sufficiency bound given in~\eqref{eq:CBP_s_Bound_Cor_sStar} and~\eqref{eq:CBP_s_Bound_Cor_sStar_Exact} with $n = 2500$, $k = 50$, $s = s^*$ and $c = 1/2$.}
	\vspace{-0.1in}
	\label{fig:cbp_exact_bounds_approx}
\end{figure}

From Fig.~\ref{fig:cbp_exact_bounds_approx}, we see that the approximations in~\eqref{eq:CBP_rho_sStar_Dominant_n}--\eqref{eq:CBP_rho_sStar_SimMag} (the red, green, and blue curves, respectively) match well with the exact expressions from Cor.~\ref{cor:CBP_s_Bound_sStar} (the black curve) under the exact recovery. Similarly, the approximation in~\eqref{eq:CBP_rho_sStar_Approx_Recovery} (the magenta dashed curve) matches well with the exact expression from Cor.~\ref{cor:CBP_s_Bound_sStar} (the magenta solid curve). A closer observation under the exact recovery condition shows that the approximation in~\eqref{eq:CBP_rho_sStar_Dominant_delta} (the green curve) closes in on the exact bound at higher $\delta$ where the condition, $\delta \gg 2/(n-k)$, is valid. Further, we observe that the testing rate has almost halved when we allow $5$ errors as compared to the exact recovery case at very high confidence, $\log(1/\delta) = 7$. \textcolor{black}{Lastly, the rate of increase of $\rho_R$ with $\log(1/\delta)$ is lower when $g_\epsilon > 0$ compared to when $g_\epsilon = 0$, i.e., the slope of the magenta curve is lower than that of the black curve in Fig.~\ref{fig:cbp_exact_bounds_approx}. For example, if $g_\epsilon = 0$, i.e., we wish to guarantee exact recovery, and if we perform $0.6 n$ tests, the set output by the algorithm will fail to match the defective set for about $15\%$ of the random test matrices. In contrast, if $g_\epsilon = 5$, with $0.5 n$ tests,  a randomly drawn test matrix will fail with probability less than $0.1\%$. As we will see, such observations hold for the other algorithms also.}

\section{PAC Analysis with False Negative Errors}
\label{sec:PAC Analysis for False Negatives}
In this section, we illustrate how PAC analysis can be used to quantify the approximation tolerance when only false negative errors occur. Specifically, we consider the DD algorithm whose defective set estimate satisfies $\hat{\mathcal{K}} \subseteq \mathcal{K}$~\cite{Aldridge_Balsassini_2014}.
\vspace{-0.3cm}
\subsection{Bernoulli Test Design: The DD Algorithm}
\label{sec:The DD Algorithm}
The DD algorithm proceeds in two stages: 1) a CoMa-like method to eliminate all the items participating in negative tests to obtain a \emph{probable defective set (PDS)} comprising the remaining items, and 2) an item in the PDS is declared a \emph{definite defective} if it is the sole item participating in any positive outcome test, after eliminating the items identified as non-defective in the first stage~\cite{Aldridge_Balsassini_2014}. Therefore, the DD algorithm never classifies a non-defective item as a defective. However, it may make false negative errors.

We characterize the number of false negatives by counting the number of \emph{unidentified defectives} i.e., the defective items that remain unidentified in the training phase due to their participation in only those group tests which have other (definite) defectives participating in them. Suppose the hypothesis, i.e., the outcome $\hat{x}(\cdot)$ of the DD algorithm, has $D \le k$ unidentified defectives, and we want the error in the testing phase to be at most $\epsilon$. This requires $D \leq d_\epsilon$, where $d_\epsilon$ can be obtained~from
\begin{align}
	{\mathbb{P}}_{\textbf{a}_i\sim \mathcal{B}(p)}\left(\hat{x}(\textbf{a}_i ) \neq x(\textbf{a}_i)\right) &= (1-(1-p)^{D})(1-p)^{k-D} \leq \epsilon \nonumber \\
	\Rightarrow d_\epsilon &= \Biggl\lfloor \frac{\log(1+\epsilon/(1-p)^k)}{\log(1/(1-p))} \Biggr\rfloor. 	\label{eq:Prob_Error_FN_Relationship}
\end{align}
Here, the PAC bound in \eqref{eq:PAC bound} reduces to $\mathbb{P}(D > d_\epsilon) \leq \delta$.

\textcolor{black}{We use a similar approach as in~\cite{Aldridge_Balsassini_2014} till~Lemma~\ref{lem:DD_Interim_Results} below, with the definition of the success probability modified to accommodate the approximate recovery condition.} Think of each group test as a \emph{ball}, and \emph{bin} the balls (group tests) into $k+2$ bins. The first $k$ bins correspond to the $k$ defectives, i.e., a ball will fall into the $i$th bin, $i = 1, 2, \ldots, k$, if only the $i$th defective item participates and no other defective item participates in the corresponding group test. The ball will fall into the $(k+1)$th bin, if the corresponding group test has more than one defective item participating. Lastly, the ball will fall into the $(k+2)$th bin, if the outcome of the corresponding test is negative, i.e., when no defectives participate.  We construct a vector $\textbf{B}^\prime = (B_1, B_2, \ldots, B_k, B_+, B_-)$ containing the number of tests of each type that are conducted in the training phase comprising $m$ independent group tests. 
If the entries of each group test are drawn i.i.d. from $\mathcal{B}(p)$, the probability vector associated with the bins is $\textbf{q} = (\underbrace{q_1, \ldots, q_1}_{k~\text{bins}},q_+,q_-),$~\cite{Aldridge_Balsassini_2014} with
\vspace*{-0.2cm}
\begin{align}
	q_1 &= \mathbb{P}(\text{Only one defective participates}) = p(1-p)^{k-1}, \nonumber \\
	q_- &= \mathbb{P}(\text{None of the defectives participates}) = (1-p)^k, \nonumber \\
	q_+ &= \mathbb{P}(\text{More than 1 defectives participate}) = 1-kq_1-q_-, \nonumber
\end{align}
and hence $\textbf{B}^\prime$ follows the multinomial distribution~\cite{Aldridge_Balsassini_2014}
\begin{equation}
	\mathbb{P}_{m;\textbf{q}}(b_1,\dots,b_k,b_+,\!b_-\!) \!= \!\frac{m!}{b_+!\ b_-!\ \prod^k_{i=1}b_i!}q_1^{b_1+\dots+b_k}q_+^{b_+}q_-^{b_-},
\end{equation}
where $\sum_{i=1}^kb_i+b_++b_- = m$. 

Recall that the PDS is a set containing all the defective items as well as $0 \le g \le n-k$ hidden non-defective items. An item participating in a positive test can be declared definite defective if and only if none of the other items in the PDS participate in that test~\cite{Aldridge_Balsassini_2014}. The probability of the event where all $g$ hidden non-defectives do not participate in a given test is $(1-p)^g$. Using this, we divide $i$th bin into two (sub-) bins containing $L_i$ and $B_i-L_i$ balls, with corresponding probabilities $(1-p)^gq_1$ and $(1-(1-p)^g)q_1$, respectively, for $i \in [k]$. We call the bin with $L_i$ balls as the $i$th singleton bin: any ball in this bin corresponds to a group test where only the $i$th defective item (and no other item in the PDS) participates. Therefore, if the $L_i \geq 1$, we can declare the $i$th item as a definite defective. Thus, the new bin vector is $\textbf{B} = (B_+, B_-, L_1, L_2, \ldots, L_k, B_1-L_1,B_2-L_2, \ldots, B_k-L_k)$ with the associated probabilities $\textbf{q} = (q_+, q_-, \underbrace{(1-p)^gq_1,\dots,(1-p)^gq_1}_{k~\text{bins}}$, $\underbrace{(1-(1-p)^g)q_1,\dots,(1-(1-p)^g)q_1}_{k~\text{bins}})$. We now have:
\begin{lemma}
	\label{lem:DD_Interim_Results}
	Consider the DD algorithm run using the outcomes of $m$ tests, with the test matrix drawn from $\mathcal{B}(p)$. Then,
\begin{enumerate}[label=(\alph*)]
	\item The probability that there are $r$ negative test outcomes is $\mathbb{P}(B_- = r) = {m \choose r}q_-^r(1-q_-)^{m-r}$.\label{lem:DD_Interim_Results_a}
	\item  If $G$ denotes the number of hidden non-defectives in the  PDS, we have $\mathbb{E}[G] = \bar{g} = (n-k)(1-p(1-p)^k)^m$. \label{lem:DD_Interim_Results_b}
	\item Given a set of $d$ unidentified defectives after the second stage, conditioned on $G=g$, we have\label{lem:DD_Interim_Results_c}\\ $\mathbb{P}(\cap_{i=1}^{d}\{L_i = 0 \} | G=g) = (1-dp(1-p)^{k-1+g})^m$. 
\end{enumerate}
\end{lemma}
We note that the relation in Lemma \ref{lem:DD_Interim_Results}~\ref{lem:DD_Interim_Results_a} gives the marginal distribution of $B_-$. Lemma \ref{lem:DD_Interim_Results}~\ref{lem:DD_Interim_Results_c} specifies the probability that a given set of $d$ singleton bins are empty, conditioned on there being $g$ hidden non-defectives. We now present a sufficiency bound on $m$ for the DD algorithm:
\begin{theorem}\label{thm:DD_Bound}
	The sufficient number of tests such that the predicted outcome based on the estimated defective set does not agree with the true outcome on the future group tests with probability at most $\epsilon$ and confidence $1-\delta$ is implicitly given by the value of $m_S$ that satisfies
	\begin{equation}
		\label{eq:DD_Bound}
		{k \choose {d_\epsilon+1}}(1-(d_\epsilon+1) p(1-p)^{k-1+\bar{g}+\tilde{g}})^{m_S} \leq \delta,
	\end{equation}
	where $d_\epsilon$ is given by~\eqref{eq:Prob_Error_FN_Relationship}, $\bar{g} = (n-k)(1-p(1-p)^k)^{m_S}$ and $\tilde{g} \ge 0$ is a parameter. 
\end{theorem}
\textcolor{black}{In the above, $\tilde{g}$ is a parameter that arises in proving the sufficiency condition in \eqref{eq:DD_Bound}. Due to the complicated form of $\mathbb{P}(G=g)$ in \eqref{eq:DD_Prob_G}, it is hard to analytically derive its precise value, but since the last expression in  \eqref{eq:DD_bound_penultimate} is monotonically increasing in $\tilde{g}$, the inequality will be satisfied for some non-negative $\tilde{g}$. In our simulations, we have seen that choosing $\tilde{g} = \lceil \bar{g} \rceil - \bar{g}$ is sufficient to ensure that the inequality \eqref{eq:DD_bound_penultimate} holds, and hence the $m_S$ obtained from \eqref{eq:DD_Bound} is indeed sufficient.}

\subsubsection{Order-Wise Analysis of~\eqref{eq:DD_Bound}}
\label{sec:DD_Order-wise_analysis}
We now discuss the behavior of $m_S$ for nonzero $\epsilon$ as $n$ and $k$ grow. We use $p = 1/k$ and $\tilde{g} = 1$ in our analysis. First, we observe that,
\begin{equation}
	\label{eq:DD_g_bar_approx}
	\bar{g} = (n-k)\left(1 - 1/ke\right)^m \to ne^{-m/ke},
\end{equation}
where we use $(1+a/x)^x \to e^{a}$ from below for large $x$ with $x = m$ and $a = -m/ke$. Using ${k \choose d_\epsilon+1} \leq \left(ek/(d_\epsilon+1)\right)^{d_\epsilon+1}$ and~\eqref{eq:DD_g_bar_approx} in~\eqref{eq:DD_bound_final}, we get
\begin{align}
	\left(\frac{ek}{d_\epsilon+1}\right)^{d_\epsilon+1} \left(1-(d_\epsilon+1)\frac{1}{k}\left(1-\frac{1}{k}\right)^{k+\bar{g}}\right)^m &\leq \delta \nonumber \\
	\left(\frac{ek}{d_\epsilon+1}\right)^{d_\epsilon+1} \left(1-\frac{d_\epsilon+1}{ke} \left(1-\frac{1}{k}\right)^{ne^{-m/ke}}\right)^m &\leq \delta \nonumber \\
	m\log\left(1\!-\!\frac{d_\epsilon\!+\!1}{ke}\left(1\!-\!\frac{1}{k}\right)^{ne^{-\!m/ke}}\right)&\\
	\label{eq:DD_asymptotic_approx_interim}
	\leq \log\left(\delta \left(\frac{d_\epsilon\!+\!1}{ke}\right)^{d_\epsilon\!+\!1}\right)&,
\end{align}
where we use the fact that $(1-1/k)^k \to 1/e$ for large $k$ along with~\eqref{eq:DD_g_bar_approx} in penultimate step. We now observe that $(1-x)^a \approx (1-ax)$ holds when $x$ is small and for fixed $a$. It also holds when $x$ becomes smaller at a \emph{faster rate} than the growth of $a$. Therefore, the condition stated above is valid for $m > ke\log n$ with $x = 1/k$ and $a = ne^{-m/ke}$. Further, for $ne^{-m/ke}/k < 1$ to hold, we require $m > ke\log(n/k)$. Since we have $m > ke \log n$, the above condition on $m$ holds too, for $k \geq 1$. Using these conditions in~\eqref{eq:DD_asymptotic_approx_interim}, we get
\begin{align}
	m\log\left(1\!-\!\frac{d_\epsilon\!+\!1}{ke}\left(1\!-\!\frac{ne^{-\!m/ke}}{k}\right)\right)\!&\leq\!\log\left(\delta \left(\frac{d_\epsilon\!+\!1}{ke}\right)^{d_\epsilon\!+\!1}\right) \nonumber \\
	\label{eq:DD_transcendental}
	m\!\left(1\!-\!\frac{ne^{-\!m/ke}}{k}\right)\!\geq\!ke&\left[\frac{\log\left(\frac{1}{\delta}\right)}{d_\epsilon\!+\!1}\!+\!\log \left(\frac{ke}{d_\epsilon\!+\!1}\right)\right]
\end{align}
where we use $\log(1-x) \to -x$ for small $x$ with $x = \frac{d_\epsilon+1}{ke}\left(1-\frac{ne^{-m/ke}}{k}\right) \ll 1$ for large $n$ and $m$ in the last step. The solution to the transcendental equation in~\eqref{eq:DD_transcendental} subject to $m > ke\log n$ gives the sufficient number of tests, $m_S$ required for the DD algorithm to succeed with confidence $1-\delta$ when $d_\epsilon$ false negative errors are allowed. Define 
\begin{equation}
	\label{eq:D}
	D \triangleq ke\left[ \frac{\log\left(\frac{1}{\delta}\right)}{d_\epsilon+1} + \log\left(\frac{ke}{d_\epsilon+1}\right) \right].
\end{equation}
Note that as $m$ increases, $D/m$ decreases. Using~\eqref{eq:D} in~\eqref{eq:DD_transcendental}, we get
\begin{align}
	- \frac{m_S}{ke} + \log\left(\frac{n}{k}\right) &\leq \log\left(1 - \frac{D}{m_S}\right) \nonumber \\
	\label{eq:DD_quadratic_transcendental}
	\frac{m_S^2}{ke} - m_S\log\left(\frac{n}{k}\right) - D &\geq 0,	
\end{align}
where the last step uses $\log(1-x) \to x$ for small $x$ with $x = D/m$. Therefore, we get
\begin{equation}
	\label{eq:DD_quadratic_solution}
	m_S = \frac{ke}{2}\left[\log\left(\frac{n}{k}\right) + \sqrt{\log^2\left(\frac{n}{k}\right) + \frac{4D}{ke}}\right].
\end{equation}
Since $m_S > ke\log n$, the solution to $m$ in~\eqref{eq:DD_quadratic_solution} is consistent. Further using $\sqrt{1+x} \approx 1+x/2$ with $x = 4D/ke\log^2(n/k)$ in~\eqref{eq:DD_quadratic_solution}, we get
\begin{align}
	\label{eq:DD_quadratic_solution_sqrtApprox}
	m_S &= \frac{ke}{2}\log\left(\frac{n}{k}\right)\left[2\!+\!\frac{2D}{ke\log^2\left(\frac{n}{k}\right)}\right]\nonumber\\
	&= ke\log\left(\frac{n}{k}\right)\!+\!\frac{D}{\log\left(\frac{n}{k}\right)}.
\end{align}
From~\eqref{eq:Prob_Error_FN_Relationship}, we see that for $p=1/k$ and large $k$, we get $d_\epsilon \approx ke\epsilon$ using $1/(1-x) \approx 1+x$, $(1-p)^k \to 1/e$ and $\log(1+x) \approx x$ for small $x$. Substituting for $D$ from~\eqref{eq:D} in~\eqref{eq:DD_quadratic_solution_sqrtApprox} along with $d_\epsilon = ke\epsilon$, we get
\begin{equation}
	\label{eq:DD_Bound_Final_Approx}
	m_S = ke\left[\log\left(\frac{n}{k}\right)\!+\!\frac{\log\left(\frac{1}{\delta}\right)}{(ke\epsilon\!+\!1)\log\left(\frac{n}{k}\right)}\!+\!\frac{\log\left(\frac{ke}{ke\epsilon+1}\right)}{\log\left(\frac{n}{k}\right)}\right].
\end{equation}

\textcolor{black}{As in the CoMa and CBP case, we see that the behavior of the DD bound for a fixed $n$, $k$ and $\delta$ is $\propto (\log(1/\epsilon) + 1/\epsilon)$ for any $\epsilon > 0$. Lastly, setting $\epsilon = 0$, we get $m_S = ke\left(\log (n/k) + \frac{1}{\log(n/k)} \left[\log(1/\delta) + \log k + 1 \right]\right)$ in the exact recovery case. Noting that $1/\log(n/k) < 1$ along with $\log k/ \log(n/k) < 1$ for $k \ll n$, we get $m_S = ke\left(\log (n/k) + \frac{\log(1/\delta)}{\log(n/k)} + 2\right)$ when $\epsilon = 0$.} Using \cite[Theorem B.3]{Aldridge_Balsassini_2014}, we get $m = (\kappa(\gamma^\prime)+\delta)ek\log n$ is sufficient, where $\kappa(\gamma^\prime) \triangleq \max \{\gamma^\prime, 1-\gamma^\prime\}$ with $k = n^{1-\gamma^\prime}$. Using $\delta^\prime = n^{-\delta}$, $m = ek(\kappa(\gamma^\prime)\log n + \log(1/\delta^\prime))$ is sufficient for exact recovery, \textcolor{black}{which is similar to our bound}.

\section{Proofs of Theorems and Lemmas}
\label{sec:Proofs of Theorems and Lemmas}
\begin{proof}[Proof of Lemma \ref{lem:defective_set_learnt_function_equivalence}]
	We first show that if $\hat{\mathcal{K}} = \mathcal{K}$, then $\hat{x}(\mathbf{a}) = x(\mathbf{a})$ with probability one. Note that there is a one-to-one mapping between any $k$-literal OR-ing function $\hat{x}$ and a corresponding $k$-sized set $\hat{\mathcal{K}}$ containing the elements participating in the OR-ing function represented by $\hat{x}$. Hence, as long as the distribution $\mathcal{D}$ is such that $\mathbb{P}_{\mathcal{D}}(a_{j} = 1) \in (0, 1),~j \in [n]$ and $a_j$s are independent,
	\begin{equation}
		 \hat{\mathcal{K}} = \mathcal{K}~~\implies~~\mathbb{P}_{\mathbf{a} \sim \mathcal{D}}\left(\hat{x}(\mathbf{a}) = x(\mathbf{a})\right) = 1.
	\end{equation}

	We show the converse by contrapositive. To this end, it suffices to show that if $\hat{\mathcal{K}} \neq \mathcal{K}$, then $\exists~\mathbf{a}$ that occurs with nonzero probability when $\mathbf{a} \sim \mathcal{D}$, such that $\hat{x}(\mathbf{a}) \neq x(\mathbf{a})$. In other words, we need to show that
	\begin{equation}
		\label{eq:Converse_condition_equivalence}
		\hat{\mathcal{K}} \neq \mathcal{K} \implies e(\hat{x}(\cdot), x(\cdot)) = \mathbb{P}_{\mathbf{a} \sim \mathcal{D}}\left(\hat{x}(\mathbf{a})\!\neq\!x(\mathbf{a})\right) > 0.
	\end{equation}
	
	Starting with $\hat{\mathcal{K}} \neq \mathcal{K}$, we argue that the above claim holds under the following covering cases:
	\begin{enumerate}
		\item When $\hat{\mathcal{K}} \setminus \mathcal{K} \neq \emptyset$, consider any $\mathbf{a} \sim \mathcal{D}$ such that
		\begin{equation}
			a_j = \begin{cases}
				1, j \in \hat{\mathcal{K}} \setminus \mathcal{K} \nonumber \\
				0, j \notin \hat{\mathcal{K}} \setminus \mathcal{K}. \nonumber
			\end{cases}
		\end{equation}
		Then, for such an $\mathbf{a}$, which occurs with nonzero probability when $\mathbf{a} \sim \mathcal{D}$, we have $\hat{x}(\mathbf{a}) = 1$ whereas $x(\mathbf{a}) = 0$.
		\item When $\mathcal{K} \setminus \hat{\mathcal{K}} \neq \emptyset$, consider any $\mathbf{a} \sim \mathcal{D}$ such that
		\begin{equation}
			a_j = \begin{cases}
				1, j \in \mathcal{K} \setminus \hat{\mathcal{K}} \nonumber \\
				0, j \notin \mathcal{K} \setminus \hat{\mathcal{K}}. \nonumber
			\end{cases}
		\end{equation}
		Then, for such an $\mathbf{a}$, which occurs with nonzero probability when $\mathbf{a} \sim \mathcal{D}$, we have $\hat{x}(\mathbf{a}) = 0$ whereas $x(\mathbf{a}) = 1$.
	\end{enumerate}
	Noting that the distribution $\mathcal{D}$ obeys $\mathbb{P}_{\mathcal{D}}(a_{j} = 1) \in (0, 1),~j \in [n]$ and that the $a_j$s are independent, we get $\mathbb{P}_{\mathbf{a} \sim \mathcal{D}}(\hat{x}(\mathbf{a}) \neq x(\mathbf{a})) > 0$ as required in~\eqref{eq:Converse_condition_equivalence} in both the cases, thereby proving the converse part.
\end{proof}

\begin{proof}[Proof of Lemma \ref{lem:probability_CoMa}]
	Recall that the probability with which an item can participate in a test is $p$. A non-defective item will remain hidden in a group test under two mutually exclusive conditions: 1) the test outcome is positive or 2) the test outcome is negative but the item does not participate in the test. \textcolor{black}{It then follows that the probability that $g$ items will be hidden, denoted by $\mathbb{P}_g^h$, can be written as
	\begin{align}
		\mathbb{P}_g^h &= \mathbb{P}[\text{positive test}] + \mathbb{P}[\text{negative test and $g$ items excluded}] \nonumber \\
					   &= (1-(1-p)^k) + (1-p)^{g+k}. \nonumber
	\end{align}
	Since the tests are independent, the probability that $g$ items will remain hidden in $m$ tests is given by
	\begin{align}
		\mathbb{P}_g^h(m) = (1 - (1-p)^k + (1-p)^{g+k})^m,
	\end{align}
	as required.}
\end{proof}

\begin{proof}[Proof of Theorem \ref{thm:CoMa_Bound}]
	If we allow at most $g_\epsilon$ non-defective items to remain hidden, an error occurs if some subset of $g_\epsilon\!+\!1$ non-defective items remains hidden. Using the union bound
	\begin{equation}
		\label{eq:CoMa_Bound_Intermediate}
		\mathbb{P}(e(\hat{x},x)>\epsilon) = \mathbb{P}(G > g_\epsilon) \leq {n-k \choose g_\epsilon + 1}\mathbb{P}^h_{g_\epsilon + 1}(m),
	\end{equation}
	where the first equality holds since CoMa only makes false positive errors. Using Lemma \ref{lem:probability_CoMa} in \eqref{eq:CoMa_Bound_Intermediate}, we get
	\begin{equation}
		\label{eq:CoMa_Bound_Final}
		\mathbb{P}(e(\hat{x},x)>\epsilon) \leq {n-k \choose g_\epsilon + 1}\left(1-(1-p)^k+(1-p)^{g_\epsilon + 1 + k}\right)^m
	\end{equation}

	Using $\mathbb{P}(e(\hat{x},x)>\epsilon) = \mathbb{P}(G > g_\epsilon) \leq \delta$ in \eqref{eq:CoMa_Bound_Final} and rearranging the terms, we get the desired result.
\end{proof}

\begin{proof}[Proof of Lemma \ref{lem:CCP_Lemma}]
	(a) Let $T$ denote the number of draws needed to collect $w-g$ distinct coupons, and let $T_i$ denote the number of draws needed to collect the $i$th coupon after $i-1$ coupons have been collected. Then, $T_1 = 1$ and $T = \sum_{i=1}^{w-g}T_i$. Now, the probability of collecting a new coupon in a single draw after $i-1$ coupons have been collected is $p_i = \frac{w-(i-1)}{w}$, so that $T_i$s are independent Geometrically distributed random variables with expectation $1/p_i$. Therefore, the expected stopping time is
	\begin{align}
		\mathbb{E}[T] &= \sum_{i=1}^{w-g} \mathbb{E}[T_i] \nonumber\\
		&= \sum_{i=1}^{w-g} \frac{1}{p_i} \nonumber \\
					  &= w\left[H_w - H_g\right] \nonumber,
	\end{align}
	where $H_w$ denotes $w$th Harmonic number as defined earlier. Using a well known asymptotic approximation for $H_w \approx \log w + \gamma$, where $\gamma \approx 0.5772$ is the Euler–Mascheroni constant as mentioned in Sec.~\ref{sec:The CBP Algorithm},\footnote{For e.g., for $w=500$, the error in the approximation is about $0.015\%$.} we get $\mathbb{E}[T] \approx w\left[\log w +\gamma - H_g\right]$, the required result.\\
	(b) Let $Z_i^r$ denote the event that the $i$th coupon was not picked in the first $r$ trials (draws). Then,
	\begin{align}
		\mathbb{P}(Z_i^r) &= \left(1-\frac{1}{w}\right)^r \nonumber \\
		&\leq e^{-r/w}. \nonumber
	\end{align}
	Let $Z_{i_1,i_2}^r$ denote the event that the $i_1$ and $i_2$th coupons, $i_1 \neq i_2$, were not picked in the first $r$ trials. Then,
	\begin{align}
		\mathbb{P}(Z_{i_1,i_2}^r) &= \left(1-\frac{1}{w}\right)^r \cdot \left(1-\frac{1}{w-1}\right)^r \nonumber \\
		&\leq e^{-r/w} \cdot e^{-r/(w-1)}. \nonumber
	\end{align}
	Continuing in this manner, the probability that $i_1, i_2, \ldots i_g$th coupons were not picked in the first $r$ trials is given by
	\begin{align}
		\mathbb{P}(Z_{\{i_l\}_{l=1}^g}^r) &= \prod_{i=0}^{g-1}\left(1-\frac{1}{w-i}\right)^r \nonumber \\
		&\leq \prod_{i=0}^{g-1} e^{-r/(w-i)} \nonumber \\
		&= e^{-r\left[H_w - H_{w-g}\right]}. \nonumber
	\end{align}
	When we want to collect only $w-g$ coupons, the event of interest occurs when $g+1$ coupons are missed in first $r = \chi w\left[\log w + \gamma - H_g\right]$ trials. Note that,
	\begin{align}
		\mathbb{P}(Z_{\{i_l\}_{l=1}^{g+1}}^r) &\leq e^{-\chi w \left[\log w + \gamma - H_g\right] \left[H_w - H_{w-(g+1)}\right]} \nonumber \\
		&\leq e^{-\chi \left[\log w + \gamma - H_g\right] (g+1)} \nonumber \\
		&= w^{-(g+1)\chi} e^{(g+1)\chi [H_g-\gamma]}, \label{eq:CCP_Lemma_IndividualEventProb}
	\end{align}
	where we use the fact that $H_w\!-\!H_{w\!-\!(g\!+\!1)}\!=\!\sum_{i=0}^g1/(w\!-\!i)\!\geq\!(g\!+\!1)/w$ in the penultimate step. Since any $g+1$ coupons out of the $w$ coupons can be missed, taking the union bound over ${w \choose g+1}$ sets, and using the inequalities ${w \choose q} \leq w^q/q!,~q \leq \sqrt{w}$ and $q! \geq q^q/e^{q-1},~q \geq 1$ with $q = g+1$ to simplify the upper bound, we obtain
	\begin{align}
		\mathbb{P}(T > \chi\mathbb{E}[T])
		&\leq {w \choose g+1} \mathbb{P}(Z_{\{i_l\}_{l=1}^{g+1}}^r) \nonumber \\
		&\leq w^{(g+1)(-\chi + 1)} \frac{e^{(g+1)\chi [H_g-\gamma] + g}}{(g+1)^{(g+1)}},
	\end{align}
	as required.	
\end{proof}
\begin{proof}[Proof of Theorem \ref{thm:CBP_s_Bound}]
	In order to obtain the tail bound on $m$ for $s$-length test vector designs, we start by modifying~\eqref{eq:CBP_CCP_Approximate_Recovery_Inequality} as follows. The right-hand side is modified to $\chi (n-k)\left[\log(n-k) + \gamma - H_{g_\epsilon}\right]$. This is because all the distinct non-defectives excluding $g_\epsilon$ have been \emph{collected} if $\chi (n-k)\left[\log(n-k) + \gamma - H_{g_\epsilon}\right]$ total non-defective items have been collected. From Lemma~\ref{lem:CCP_Lemma} (b), the probability that the stopping time is more than $\chi$ times the expected stopping time is at most $(n-k)^{(g_\epsilon+1)(-\chi+1)}\frac{e^{(g_\epsilon+1)\chi [H_{g_\epsilon}-\gamma]+g_\epsilon}}{(g_\epsilon+1)^{(g_\epsilon+1)}}$.
	
	The left-hand side of~\eqref{eq:CBP_CCP_Approximate_Recovery_Inequality} is multiplied with $(1-\eta)$, where $\eta$ is a design parameter to be specified by the Chernoff bound. Then, the probability that the actual number of items in the negative tests is smaller than $(1-\eta)$ times the expected number, $ms((n-k)/n)^s$ is at most~\cite{Chan_Jaggi_Saligrama_Agnihotri_2014}
	\begin{equation}
		\exp\left( -\eta^2 m \left( \frac{n-k}{n} \right)^s \right).
	\end{equation}
	
	Taking the union bound over the above two low probability events, we have that
	\begin{equation}
		\label{eq:CBP_CCP_Approximate_Recovery_Inequality_Rho_Chi}
		(1-\eta)ms\left(\frac{n-k}{n}\right)^s \geq \chi (n-k) \left[ \log(n-k) + \gamma - H_{g_\epsilon} \right]
	\end{equation}
	does not hold with probability
	\begin{equation}
		\label{eq:CBP_s_length_Failure_Probability_Raw_Form}
		\exp\left(\!-\!\eta^2 m \left( \frac{n\!-\!k}{n} \right)^s \right)\!+\!(n\!-\!k)^{(g_\epsilon\!+\!1)(\!-\!\chi\!+\!1)}\frac{e^{(g_\epsilon\!+\!1)\chi [H_{g_\epsilon}\!-\!\gamma]\!+\!g_\epsilon}}{(g_\epsilon\!+\!1)^{(g_\epsilon\!+\!1)}}.
	\end{equation}
	From~\eqref{eq:CBP_CCP_Approximate_Recovery_Inequality_Rho_Chi}, we have
	\begin{equation}
		\label{eq:Bound_m_s_length_Intermediate}
		m \geq \frac{\chi(n-k)}{(1-\eta)s\left(\frac{n-k}{n}\right)^s} \left[\log(n-k)+\gamma-H_{g_\epsilon}\right].
	\end{equation}
	Substituting~\eqref{eq:Bound_m_s_length_Intermediate} in ~\eqref{eq:CBP_s_length_Failure_Probability_Raw_Form}, $P_e$ is upper bounded by
	\begin{align}
		P_e &\leq e^{-\eta^2 m \left( \frac{n-k}{n} \right)^s} \nonumber\\ &+ \frac{e^{(g_\epsilon + 1)(-\chi + 1)\log(n-k) + (g_\epsilon + 1)\chi [H_{g_\epsilon}-\gamma] + g_\epsilon}}{e^{(g_\epsilon + 1)\log(g_\epsilon + 1)}} \nonumber \\
		&\leq \exp\left( \frac{-\eta^2 \chi (n-k)}{(1-\eta)s} \left[\log(n-k)+\gamma-H_{g_\epsilon}\right] \right)\nonumber\\
		\label{eq:CBP_CCP_Pe_Split_Bound}
		&+ e^{-\chi (g_\epsilon+1)\left[\log(n-k)+\gamma-H_{g_\epsilon}\right] + (g_\epsilon+1)\log\left( \frac{n-k}{g_\epsilon+1} \right)+g_\epsilon}
	\end{align}
	Since $P_e = P(G > g_\epsilon) \leq \delta \in (0,1)$ and each term in the RHS of~\eqref{eq:CBP_CCP_Pe_Split_Bound} is greater than $0$, we choose a design parameter $c \in (0,1)$ such that 
	\begin{enumerate}[label=(\alph*)]
		\item $\exp\left( \frac{-\eta^2\chi(n-k)}{(1-\eta)s} \left[\log(n-k)+\gamma-H_{g_\epsilon}\right] \right) \leq (1-c)\delta$, and
		\item $e^{-\chi (g_\epsilon+1)\left[\log(n-k)+\gamma-H_{g_\epsilon}\right] + (g_\epsilon+1)\log\left( \frac{n-k}{g_\epsilon+1} \right) + g_\epsilon} \leq c\delta$.
	\end{enumerate}	
	Simplifying (b), we get
	\begin{align}
		\label{eq:CBP_CCP_Chi_Bound}
		 &\geq \frac{\left[\frac{\log\left(\frac{1}{c\delta}\right)}{g_\epsilon+1} + \frac{g_\epsilon}{g_\epsilon+1} + \log\left( \frac{n-k}{g_\epsilon+1} \right)\right]}{\log(n-k)+\gamma-H_{g_\epsilon}}.
	\end{align}
	Substituting~\eqref{eq:CBP_CCP_Chi_Bound} in (a), we get
	\begin{align}
		\frac{\eta^2}{1-\eta}\left(\frac{n-k}{s}\right) \left[ \frac{\log\left(\frac{1}{c\delta}\right)}{g_\epsilon+1}+\frac{g_\epsilon}{g_\epsilon+1}+\log\left( \frac{n-k}{g_\epsilon+1} \right) \right] \nonumber\\ \geq \log\left(\frac{1}{(1-c)\delta}\right) \nonumber \\
		\label{eq:CBP_CCP_Rho_Quadratic}
		\frac{\!\eta^2}{1\!-\!\eta} \geq C \triangleq \frac{\log\left(\frac{1}{(1-c)\delta}\right)}{\left(\frac{n\!-\!k}{s}\right) \left[ \frac{\log\left(\frac{1}{c\delta}\right)}{g_\epsilon\!+\!1}\!+\!\frac{g_\epsilon}{g_\epsilon\!+\!1}\!+\!\log\left( \frac{n\!-\!k}{g_\epsilon\!+\!1} \right) \right]}.
	\end{align}
	Noting that $C > 0$, the bound on $\eta$ is obtained by solving the  inequality $\eta^2 + C\eta - C \geq 0$ subject to $\eta \in (0,1)$ to get
	\begin{equation}
		\label{eq:CBP_CCP_Rho_Bound}
		\eta \geq \frac{-C + \sqrt{C^2 +4C}}{2}.
	\end{equation}
 Finally,~\eqref{eq:Bound_m_s_length_Intermediate},  with~\eqref{eq:CBP_CCP_Chi_Bound}
 and~\eqref{eq:CBP_CCP_Rho_Bound}, gives the required bound for the sufficient number of tests for the CBP algorithm.
\end{proof}

\begin{proof}[Proof of Corollary \ref{cor:CBP_s_Bound_sStar}]
	We follow on similar lines as~\cite{Chan_Jaggi_Saligrama_Agnihotri_2014} and simplify~\eqref{eq:CBP_s_Bound} at $s = s^*$, to get
	\begin{align}
		m &\geq \frac{\chi}{1\!-\!\eta} \frac{n\!-\!k}{\left(\frac{n\!-\!k}{n}\right)^{s^*}} \log\left(\frac{n}{n\!-\!k}\right) \left[\log(n-k)\!+\!\gamma\!-\!H_{g_\epsilon}\right] \nonumber \\
		\label{eq:Bound_m_s_length_Modified}
		&\geq \frac{\chi}{1\!-\!\eta} \left(\frac{n}{n\!-\!k}\right)^{s^*}\!\left( k\!-\! \frac{k^2}{2(n\!-\!k)} \right) \left[\log(n\!-\!k)\!+\!\gamma\!-\!H_{g_\epsilon}\right] \\
		\label{eq:Bound_m_s_length_Simplified}
		&\geq \frac{\chi k}{1\!-\!\eta} \left(\frac{n}{n\!-\!k}\right)^{s^*} \left[\log(n\!-\!k)\!+\!\gamma\!-\!H_{g_\epsilon}\right]. 
	\end{align}
	where~\eqref{eq:Bound_m_s_length_Modified} is obtained by using $\log(1+x) \geq x-x^2/2$ with $x = k/(n\!-\!k)$ and~\eqref{eq:Bound_m_s_length_Simplified} holds because increasing the RHS of~\eqref{eq:Bound_m_s_length_Modified} can only make the error probability in~\eqref{eq:CBP_s_length_Failure_Probability_Raw_Form} smaller.
\end{proof}

\begin{proof}[Proof of Lemma \ref{lem:DD_Interim_Results}]
Part \ref{lem:DD_Interim_Results_a} follows because the marginals are binomial, so it only remains to show parts \ref{lem:DD_Interim_Results_b} and \ref{lem:DD_Interim_Results_c}. 	

\ref{lem:DD_Interim_Results_b} A non-defective item will be in the PDS if it does not participate in any of the negative tests. If there are $B_-$ negative tests, the probability that any given item among the $n-k$ non-defective items will be in the PDS is $(1-p)^{B_-}$. Since the tests are independent, denoting the number of hidden non-defectives by the random variable $G$, it is easy to see that
	\begin{equation}
		\label{eq:DD_G_given_B_minus}
		G|B_- \sim \text{Bin}(n-k,(1-p)^{B_-}).
	\end{equation}
	Using $\mathbb{E}[G|B_-=r] = (n-k)(1-p)^r$ \textcolor{black}{, $0 \le r \le m$}, we get
	\begin{align}
		\bar{g}
		=& \sum_{g=0}^{n-k} g\sum_{r=0}^{m}\mathbb{P}(B_-=r)\mathbb{P}(G=g|B_-=r) \nonumber \\
		=& \sum_{r=0}^{m}\mathbb{P}(B_-=r)\mathbb{E}[G|B_-=r] \nonumber \\
		=& (n-k)\sum_{r=0}^{m}{m \choose r}q_-^r(1-q_-)^{m-r}(1-p)^r \nonumber \\
		=& (n-k)(1-p(1-p)^k)^m, \nonumber
	\end{align}
	where the last step follows by using part (a) of this lemma and the binomial expansion for $(a+b)^m$.
	
	\ref{lem:DD_Interim_Results_c} $\mathbb{P}(\cap_{i=1}^{d}\{L_i = 0\} | G=g)$ denotes the probability that  the set output by the DD algorithm misses $d$ of the defective items, conditioned on $g$ hidden non-defective items being present in the PDS. In a single test, a defective item will \emph{not} be missed, i.e., it will be identified as a definite defective, if it is the sole item among the PDS participating in that test; this occurs with probability $p(1-p)^{k-1}(1-p)^g$. Since the $d$ items being classified as definite defectives are mutually exclusive events, the probability that none of the $d$ defective items are correctly classified as a definite defective in a single test is $1 - dp(1-p)^{k-1}(1-p)^g$. Finally, this event should happen across all the $m$ independently drawn tests. Thus, we have that
	\begin{equation}
		\mathbb{P}(\cap_{i=1}^{d}\{L_i = 0\} | G=g) = (1 - dp(1-p)^{k-1+g})^m, \nonumber
	\end{equation}
which completes the proof. 
\end{proof}

\begin{proof}[Proof of Theorem \ref{thm:DD_Bound}]
	The error event will occur if \emph{more than} $d_\epsilon$ defective items remain unidentified, i.e., $\cap_{i=1}^{d_\epsilon + 1}\{L_i = 0\}$ occurs. Using Lemma~\ref{lem:DD_Interim_Results}~\ref{lem:DD_Interim_Results_c} and the union bound, we get
	\begin{equation}
		\label{eq:DD_Conditional_Prob}
		\mathbb{P}(e(\hat{x},x)>\epsilon|G=g) \!\leq\! {k \choose d_\epsilon + 1}(1-(d_\epsilon\!+\!1)p(1\!-\!p)^{k\!-\!1\!+\!g})^m,
	\end{equation}
	and hence 
	\begin{align}
		\mathbb{P}&(e(\hat{x},x)\!>\!\epsilon) \nonumber\\
		&\leq \sum_{g=0}^{n\!-\!k} {k \choose d_\epsilon\!+\!1}(1\!-\!(d_\epsilon\!+\!1)p(1\!-\!p)^{k\!-\!1\!+\!g})^m \mathbb{P}(G = g). \label{eq:pe_total}
	\end{align}
	Now, in order to characterize $\mathbb{P}(G = g)$, note that since the tests are drawn independently, $B_- \sim~\text{Bin}(m, q_-),$ where $q_- = (1-p)^k$. Using this along with the fact that $G|B_- \sim \text{Bin}(n-k,(1-p)^{B_-})$ (see ~\eqref{eq:DD_G_given_B_minus}), we get
	\begin{align}
		\mathbb{P}(G = g) &= \sum_{b=0}^{m}~\mathbb{P}(G = g | B_- = b) \mathbb{P}(B_- = b) \nonumber \\
						  &= \sum_{b=0}^{m} \left[{n-k \choose g} (1-p)^{bg} (1-(1-p)^b)^{n\!-\!k\!-\!g} \right. \nonumber\\ &\left. \times {m \choose b} (1-p)^{kb} (1-(1-p)^k)^{m\!-\!b} \right] \nonumber \\
						  &= {n-k \choose g} \sum_{b=1}^{m} \left[ {m \choose b} (1-p)^{b(g+k)} \right. \nonumber\\
						  \label{eq:DD_Prob_G}
						  &\left.\times (1-(1-p)^b)^{n\!-\!k\!-\!g} (1-(1-p)^k)^{m\!-\!b} \right].
	\end{align}
Note that, in \eqref{eq:DD_Prob_G}, we sum the terms from $b=1$ since $1-(1-p)^b = 0$ when $b=0$. It is easy to see that $m{n-k \choose g} {m \choose m/2} (1-p)^{(g+k)}$ is a loose upper bound on $\mathbb{P}(G=g)$. Clearly, for any given $m$ and large enough $g$, $\mathbb{P}(G=g)$ becomes negligible. Using this to replace $g$ by $\bar{g}+\tilde{g}$ for some non-negative $\tilde{g}$ (with $\bar{g}$ as given by Lemma~\ref{lem:DD_Interim_Results}~\ref{lem:DD_Interim_Results_b}), we get
	\begin{align}
		\mathbb{P}&(e(\hat{x},x)>\epsilon) \nonumber\\
		&\leq {k \choose 	d_\epsilon\!+\!1}\sum_{g=0}^{n\!-\!k} (1\!-\!(d_\epsilon\!+\!1)p(1\!-\!p)^{k\!-\!1\!+\!g})^m \mathbb{P}(G = g)  \\	
				&= {k \choose d_\epsilon\!+\!1}\left\{\sum_{g=0}^{\bar{g}+\tilde{g}} (1\!-\!(d_\epsilon\!+\!1)p(1\!-\!p)^{k\!-\!1\!+\!g})^m \mathbb{P}(G\!=\!g) \right.  \\
				&~~~~~+ \left.\sum_{g=\bar{g}+\tilde{g}+1}^{n\!-\!k}\!(1\!-\!(d_\epsilon\!+\!1)p(1\!-\!p)^{k\!-\!1\!+\!g})^m \mathbb{P}(G\!=\!g) \right\} \\
				&\leq {k \choose d_\epsilon\!+\!1} (1\!-\!(d_\epsilon\!+\!1)p(1\!-\!p)^{k\!-\!1\!+\!\bar{g}+\tilde{g}})^m \sum_{g=0}^{\bar{g}+\tilde{g}} \mathbb{P}(G\!=\!g) \label{eq:DD_bound_penultimate} \\
		 &\leq {k \choose d_\epsilon\!+\!1}(1-(d_\epsilon + 1)p(1-p)^{k-1+\bar{g}+\tilde{g}})^m, \label{eq:DD_bound_final}
	\end{align} 
	where $\tilde{g} \ge 0$ is a tuning parameter chosen such that the inequality in~\eqref{eq:DD_bound_penultimate} holds.
\end{proof}

\section{Simulation Results} 
\label{sec:Simulation Results}

\subsection{Tightness of the Bounds}
\label{sec:Tightness of the Bounds}
We simulate exact and approximate set identification scenarios for the CoMa, CBP, and DD algorithms and compare them with our theoretical bounds and the existing results \cite{Chan_Jaggi_Saligrama_Agnihotri_2014, Aldridge_Balsassini_2014}. We consider $n = 2500$ items, out of which $k = 50$ are defective. Further, we use Bernoulli parameter $p = 1/k$ for generating the plots with CoMa and DD algorithms and $s = s^*$ with the weight parameter, $c = 1/2$, in the CBP algorithm (see Theorem~\ref{thm:CBP_s_Bound}). The simulated curves are obtained by averaging over $1,\!000$ Monte Carlo runs.

\begin{figure*}[t]
	\centering
	\hspace*{-0.1in}
	\includegraphics[width=0.95\linewidth]{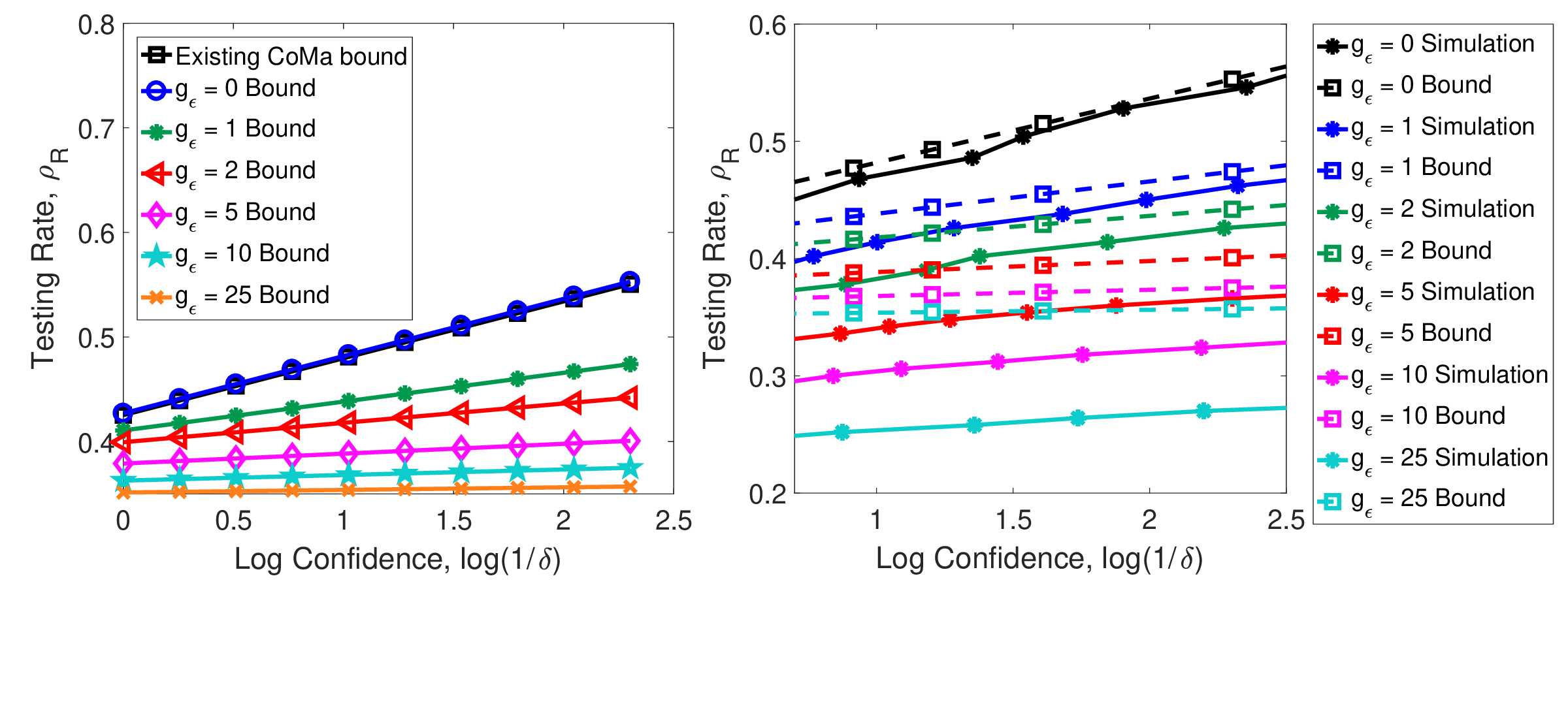}
	\vspace{-0.6in}
	\caption{(Left) Comparison of the sufficiency bound in~\cite[Thm. 4]{Chan_Jaggi_Saligrama_Agnihotri_2014} and theoretical PAC bounds \eqref{eq:CoMa_Bound} on the testing rate; (Right) theoretical and simulated testing rates at different error tolerance values, for the CoMa algorithm.}
	\vspace{-0.1in}
	\label{fig:simulation_coma_plots}
\end{figure*}
\begin{figure*}[h!]
	\centering
	\hspace*{-0.1in}
	\includegraphics[width=0.95\linewidth]{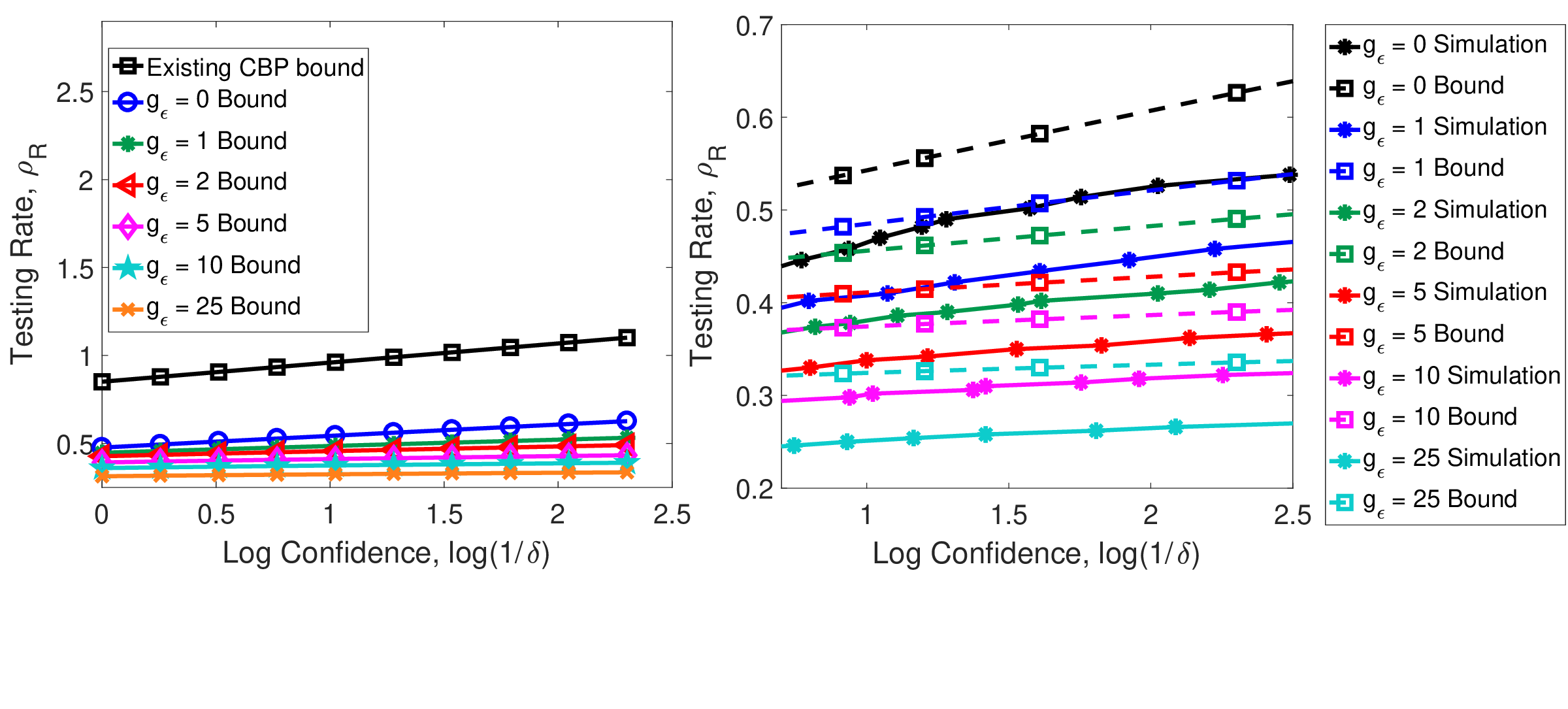}
	\vspace{-0.6in}
	\caption{(Left) Comparison of the sufficiency bound in~\cite[Thm. 3]{Chan_Jaggi_Saligrama_Agnihotri_2014} and the theoretical PAC bounds \eqref{eq:CBP_s_Bound} on the testing rate; (Right) theoretical and simulated testing rates at different error tolerance values, for the CBP algorithm.}
	\vspace{-0.1in}
	\label{fig:simulation_cbp_plots}
\end{figure*}
\begin{figure*}[h!]
	\centering
	\hspace*{-0.1in}
	\includegraphics[width=0.95\linewidth]{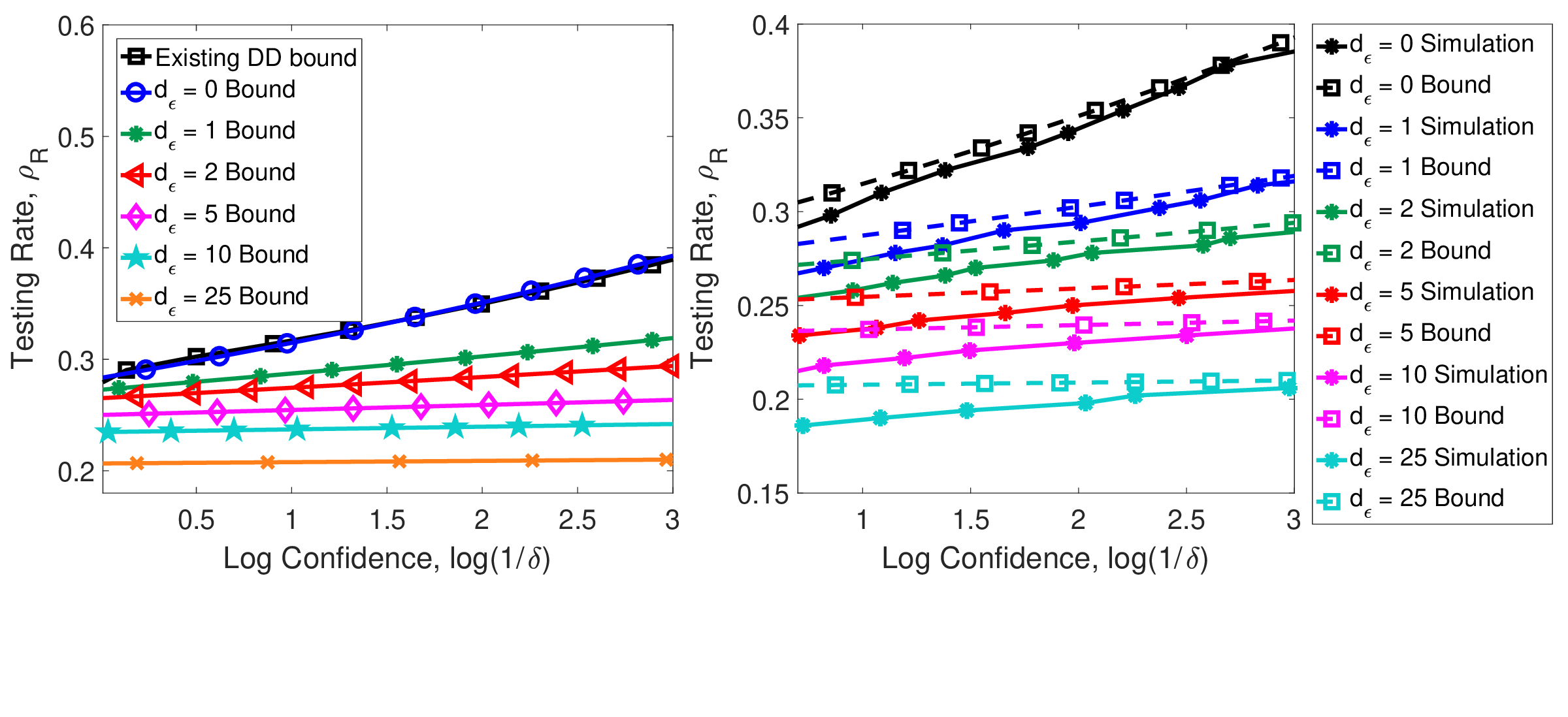}
	\vspace{-0.6in}
	\caption{(Left) Comparison of the sufficiency bound in~\cite[Lem. A.8]{Aldridge_Balsassini_2014} and the theoretical PAC bounds \eqref{eq:DD_Bound} on the testing rate; (Right) theoretical and simulated testing rates at different error tolerance values, for the DD algorithm.}
	\vspace{-0.1in}
	\label{fig:simulation_dd_plots}
\end{figure*}

Figures~\ref{fig:simulation_coma_plots},~\ref{fig:simulation_cbp_plots} and~\ref{fig:simulation_dd_plots} show the relationship between the testing rate, $\rho_R = m_S/n$ and $\log(1/\delta)$, a notion similar in behavior to the confidence parameter, $1-\delta$, for CoMa, CBP, and DD algorithms, respectively. In particular, 1) the \emph{left} subplots illustrate how the proposed PAC bound compare with the existing bounds whereas 2) the \emph{right} subplots compare the theoretical PAC bounds with the simulation curves.

From the left subplots of Figs.~\ref{fig:simulation_coma_plots} and~\ref{fig:simulation_dd_plots}, we note that the testing rate obtained from the PAC analysis equations when $\epsilon = 0$ (blue curves) match with the existing upper bounds under exact recovery (black curves) for both CoMa and DD. In contrast, from left subplot of Fig.~\ref{fig:simulation_cbp_plots}, we note that there is a difference of $\approx 0.4-0.5$ in testing rates between the PAC bound when $g_\epsilon = 0$ (blue curve) and the existing sufficiency condition under the exact recovery (black curve). This difference can be attributed to the choice of the Chernoff parameter $\eta$ in the two cases. The authors in~\cite{Chan_Jaggi_Saligrama_2011} choose $\eta=1/2$, whereas, we optimize $\eta$ to obtain the minimum number of tests such that the probability of error is lower than $\delta$. \textcolor{black}{We defer further discussion on optimizing $\eta$ to Sec.~\ref{sec:Effect of Optimizing Eta in CBP}.} We note that the sufficient number of tests is lower when we allow $\epsilon > 0$ in all three algorithms. Also, the right subplots validate the PAC upper bound derived in the paper, and we note that the bound is particularly tight in the case of the DD algorithm (Fig.~\ref{fig:simulation_dd_plots}.) 

We see that allowing for a few missed/false positive items in the set output by the algorithm can help significantly reduce the sufficient number of tests needed. As mentioned earlier in Sec.~\ref{sec:Approximate Defective Set Identification}, allowing for a few false positives (as in the CoMa and CBP algorithms) can also be useful in applications where the goal is to identify most of the non-defective items \cite{Sharma_Murthy_2017}, since these algorithms do not miss defective items. Similarly, as in the rare antigen identification example mentioned in Sec.~\ref{sec:Approximate Defective Set Identification}, allowing for a small number of missed defectives is useful when it is important to quickly identify some of the defective items, and we see that by not requiring that \emph{all} the defective items be found, the DD algorithm can significantly reduce the number of group tests that need to be conducted, while retaining a high confidence in the outcome.

\textcolor{black}{Lastly, we observe from the right subplots of Figs.~\ref{fig:simulation_coma_plots} and~\ref{fig:simulation_cbp_plots} that the CoMa bound is tight for exact recovery whereas the CBP bound follows the slope of the simulated curves slightly better at higher $g_\epsilon$ and lower confidence.}

\subsection{Effect of Approximation Error Tolerance on the Bounds}
\label{sec:Effect of Approximation Error}
In this subsection, we discuss the achievability bounds for the three algorithms, CoMa, CBP, and DD, as a function of the number of errors allowed. Fig.~\ref{fig:approx_error} shows how the testing rate, $\rho_R = m_S/n$, of CoMa and CBP (left subplot) and DD (right subplot) varies as the number of FPs and FNs, respectively, for different values of the parameter $\delta$. In the plots, the $m_S$ is computed by setting $n = 2500$, $k = 50$ using Theorem~\ref{thm:CoMa_Bound}, Theorem~\ref{thm:CBP_s_Bound}, and Theorem~\ref{thm:DD_Bound} for CoMa with $p = 1/k$, CBP with $s = s^*$, and DD with $p = 1/k$, respectively.

From Fig.~\ref{fig:approx_error}, we see that the testing rate, $\rho_R$, (and thus the sufficiency bound on $m$) is proportional to $(\log(1/\tau) + 1/\tau)$, where $\tau = g_\epsilon+1$ in the CoMa and CBP cases, and $\tau = d_\epsilon+1$ in the case of DD. From the left subplot, we see that although the testing rate for CBP as computed from our bounds for $g_\epsilon = 0$ is higher than for CoMa, the slope for CBP is steeper as $g_\epsilon$ increases. Therefore, the testing rate obtained from the CBP bound at, say, $g_\epsilon = \{25, 30\}$ is lower than that obtained from the CoMa bound. For example, at $g_\epsilon = 30$, the testing rate of CoMa is $\rho_R \approx 0.3574$ and that of CBP is $\rho_R \approx 0.325$ at a $\delta = 0.1$. In practice, both CoMa and CBP are functionally the same. However, the differences seen can be attributed to the analysis procedure \textcolor{black}{and the way the test matrix is constructed.} A similar observation can be made from the right subplots of Fig.~\ref{fig:simulation_coma_plots} and Fig.~\ref{fig:simulation_cbp_plots}. From the right subplots of Fig.~\ref{fig:simulation_coma_plots} and Fig.~\ref{fig:simulation_cbp_plots}, we see that the CoMa bound matches very well with the simulation curve at $g_\epsilon = 0$ whereas the CBP bound is relatively loose. On the other hand, the CoMa bound is looser than the CBP bound relative to their respective simulation curves at $g_\epsilon = 25$. Thus, the analysis used in deriving the CoMa bound is useful for low approximation error scenarios, whereas that used in the CBP is useful when we allow for higher approximation error. This can be attributed to the use of the union-bound argument (see~\eqref{eq:CoMa_Bound_Intermediate}). The slope of the factor ${n-k \choose g_\epsilon+1}$ is bounded between $[(c_l/x)^x (\log(c_l/x) -1), (c_u/x)^x (\log(c_u/x) -1)$ where $x = g_\epsilon+1$, $c_u = e(n-k)$ and $c_l = n-k$. That is, the slope increases with $g_\epsilon$, making the CoMa bound looser relative to the CBP bound at higher $g_\epsilon$ when $g_\epsilon+1 \ll (n-k)/2$.
	
Further, from the right subplot of Fig.~\ref{fig:approx_error}, the testing rate of the DD algorithm is lower than that of CoMa and CBP algorithms, due to the dependency of $m_S$ on $\kappa(\gamma^\prime)\log n = \log k$ in the former case as against a dependency on $\log n$ in CoMa/CBP.
\begin{figure}[t]
	\centering
	\hspace*{-0.1in}
	\includegraphics[width=0.8\linewidth]{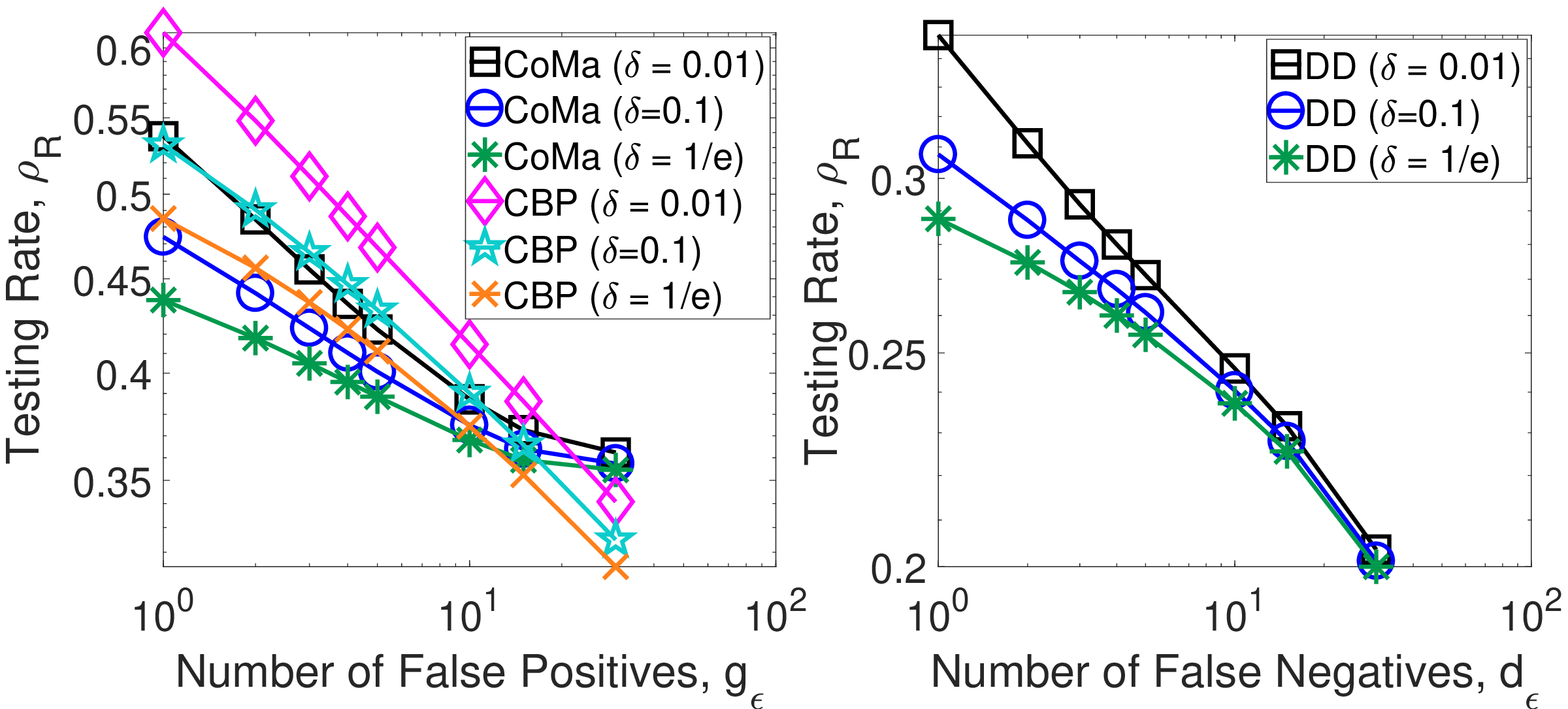}
	\caption{The sufficiency bound for CoMa and CBP (left subplot) and DD (right subplot) as the number of errors allowed varies, for different values of~$\delta$.}
	\vspace{-0.1in}
	\label{fig:approx_error}
\end{figure}

\subsection{Effect of Performing a Insufficient Number of Tests}
\label{sec:Effect of Performing Insufficient Number of Tests}
In this subsection, we present the utility of our bounds from an alternative viewpoint: what guarantees can be provided for a given $n$ and $k$, if the number of group tests performed is insufficient to guarantee exact recovery with high confidence? To illustrate this, in Fig.~\ref{fig:delta_eps_curves}, we plot the parameter, $\delta$, as a function of $g_\epsilon$ for CoMa and CBP algorithms, and as a function of $d_\epsilon$ for DD algorithm, at two different values of $m$ for each algorithm. The plot is generated by setting $n = 2500$, $k = 50$, $p = 1/k$, $c = 1/2$, $s = s^*$ in Theorems~\ref{thm:CoMa_Bound},~\ref{thm:CBP_s_Bound} and~\ref{thm:DD_Bound} for the purpose of discussion. However, similar observations can be made across different values of $n$, $k$, $m$ etc. 
\begin{figure}[t]
	\centering
	\includegraphics[width=0.8\linewidth]{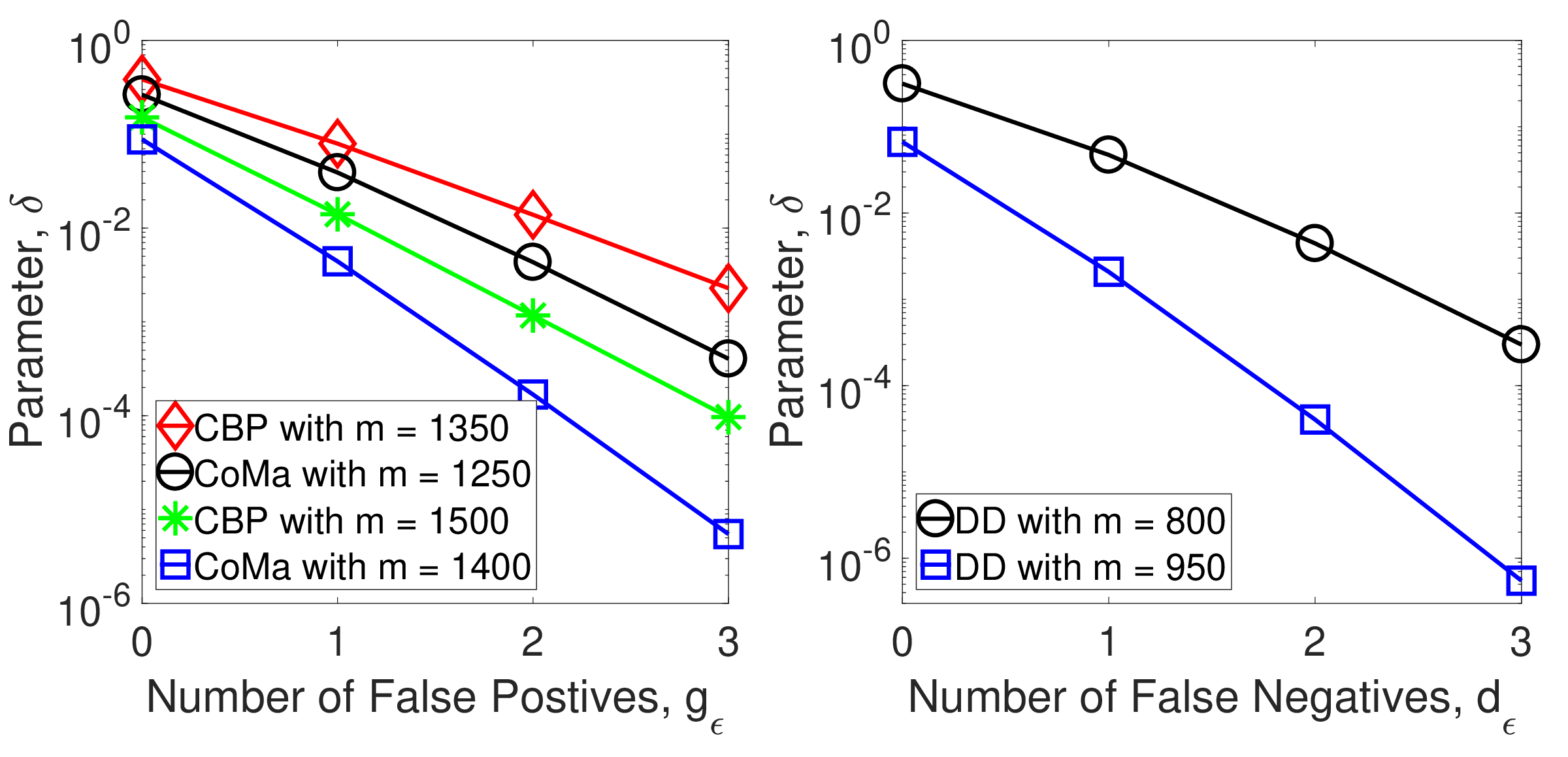}
	\caption{Illustration of the trade-off between the confidence parameter and the approximation error tolerance for CoMa, CBP and DD algorithms.}
	\vspace{-0.1in}
	\label{fig:delta_eps_curves}
\end{figure}

From Fig.~\ref{fig:delta_eps_curves}, we see that allowing a small number of false positive/negative errors allows one to obtain significantly higher confidence, i.e., a lower $\delta$, for a given number of tests. For example, to ensure a confidence of $\approx 91\%$, with the CoMA algorithm, $1400$ tests are sufficient as per~\eqref{eq:CoMa_Bound} with $n = 2500$, $k = 50$ and $p = 1/k$ for exact recovery. However, if only $1250$ tests could be conducted, one can provide a confidence of $\approx 73\%$ for exact recovery, while tolerating a single false positive error yields a confidence of $\approx 96\%$ ($> 91\%$). If two false positives are allowed, the confidence goes well above~$99\%$. Similar conclusions can be drawn about the CBP and DD algorithms.

This example illustrates that the PAC bounds can provide a range of guarantees when the number of group tests performed is insufficient to guarantee exact recovery with high confidence. One can either tolerate a small number of errors or choose to operate at lower confidence levels.

\begin{figure}[t]
	\centering
	\hspace*{-0.1in}
	\includegraphics[width=0.8\linewidth]{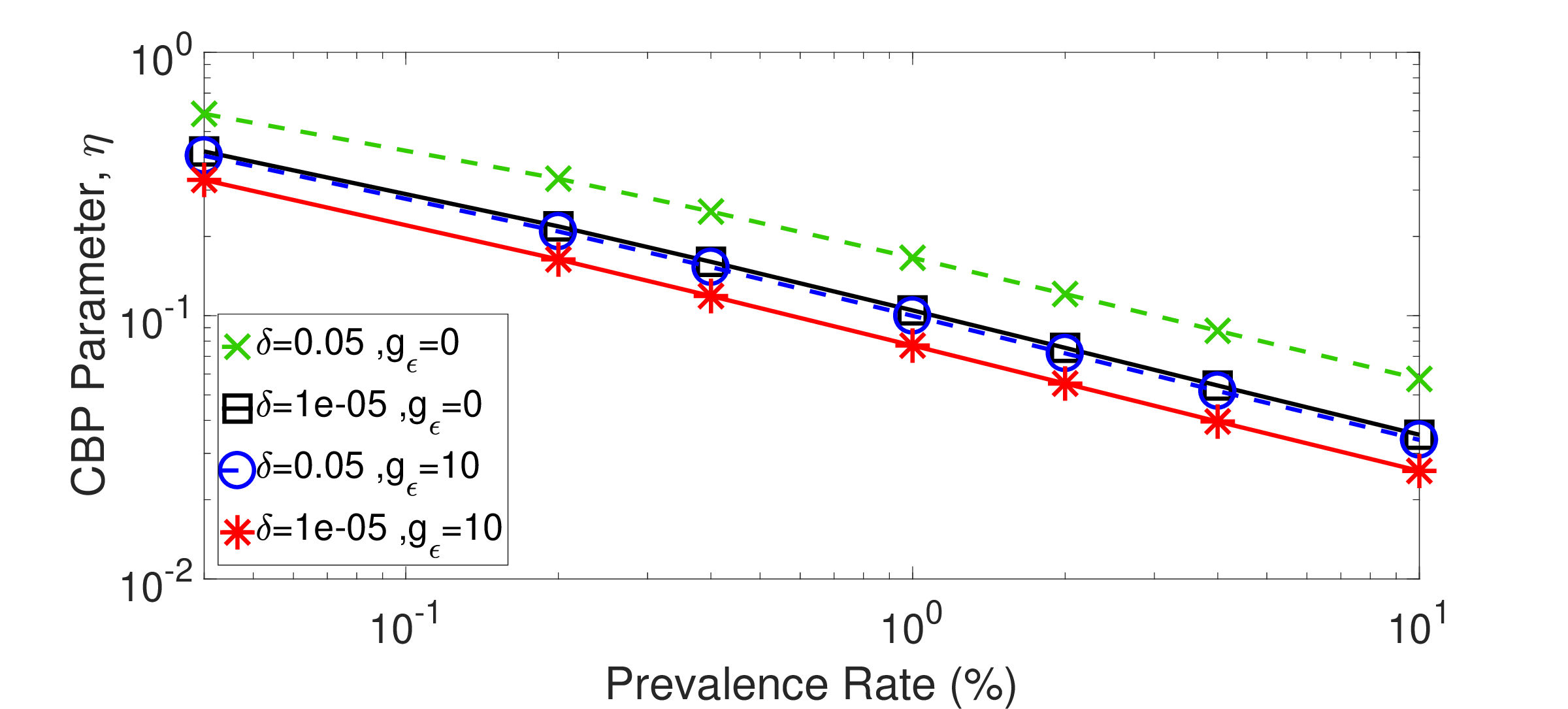}
	\caption{The CBP parameter, $\eta$, computed at different $g_\epsilon$ and $\delta$ with $n = 2500$, $c = 1/2$, $s = s^*$, as the prevalence rate (\%) is varied.}
	\vspace{-0.1in}
	\label{fig:eta_vs_p}
\end{figure}

\subsection{Effect of Optimizing $\eta$ in CBP}
\label{sec:Effect of Optimizing Eta in CBP}
In this subsection, we present empirical evidence that optimizing the Chernoff parameter $\eta$ yields tighter bounds on the number of tests. For this experiment, we set $n = 2500$, $s = s^*$ and $c = 1/2$ in the CBP algorithm (see Theorem~\ref{thm:CBP_s_Bound}). From~\eqref{eq:CBP_s_Bound}, it is clear that $m_S$ decreases as $\eta$ is decreased. At the same time, choosing an $\eta$ such that~\eqref{eq:CBP_CCP_Pe_Split_Bound} is upper bounded by $\delta$ ensures that the sufficiency condition is satisfied. We illustrate the effect of varying the prevalence rate ($\%$), i.e., $100 \times k/n$ at $\delta \in \{10^{-5}, 0.05\}$ and $g_\epsilon \in \{0, 10\}$ in Fig.~\ref{fig:eta_vs_p}. We see that $\eta$ computed using~\eqref{eq:CBP_CCP_Rho_Bound} decreases as the prevalence rate increases irrespective of the choice of $\delta$ and $g_\epsilon$. For a fixed $\delta$, the value of $\eta$ is higher under exact recovery (i.e., when $g_\epsilon = 0$) to that when $g_\epsilon = 10$, and, consequently, the sufficient number of tests is lower in the latter case. For a fixed $g_\epsilon$, as the confidence, $1-\delta$, decreases, $\eta$ decreases, thereby lowering the sufficient number of tests. The gap between the values of $\eta$ is higher across the above $\delta$ and $g_\epsilon$ values when the prevalence rate is lower, in fact, close to $0\%$, and as we move towards right in Fig.~\ref{fig:eta_vs_p}, the gap decreases.

\begin{figure}[t]
	\centering
	\hspace*{-0.1in}
	\includegraphics[width=0.8\linewidth]{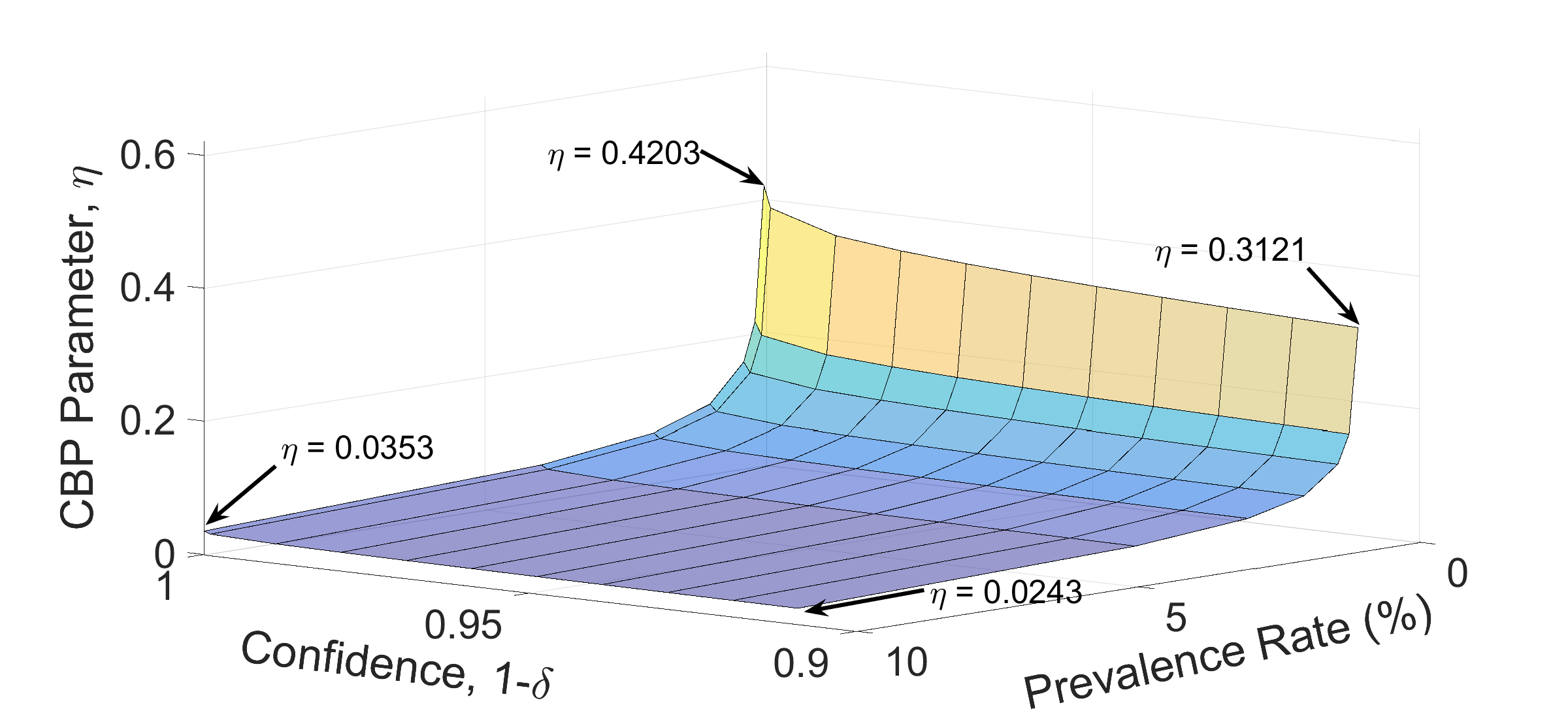}
	\caption{The surface of the CBP parameter, $\eta$, computed at $g_\epsilon = 0$ with $n = 2500$, $c = 1/2$, $s = s^*$, as the prevalence rate (\%) and $\delta$ are varied.}
	\vspace{-0.1in}
	\label{fig:eta_surface}
\end{figure}

Next, in Fig.~\ref{fig:eta_surface}, we show the effect of varying $\delta$ and the prevalence rate using a surface plot of $\eta$, for $g_\epsilon = 0$. When the prevalence rate is very low, $\eta = 1/2$ is optimum. Otherwise, the optimum value of $\eta$ is much lower. Hence, computing $\eta$ using~\eqref{eq:CBP_CCP_Rho_Bound} leads to a tighter sufficiency bound, an improvement by a factor of $2$ (see left subplot of Fig.~\ref{fig:simulation_cbp_plots}).

We conclude our discussion on the parameter $\eta$ by empirically showing its variation as $n$ varies in Fig.~\ref{fig:eta_vs_n}. The optimum value of $\eta$ decreases as $n$ increases. The variation across the prevalence rate remains similar to the observations made earlier. Therefore, in conclusion, the computation of the optimum Chernoff parameter, $\eta$, can significantly improve the sufficiency bound on $m$.
\begin{figure}[t]
	\centering
	\hspace*{-0.1in}
	\includegraphics[width=0.8\linewidth]{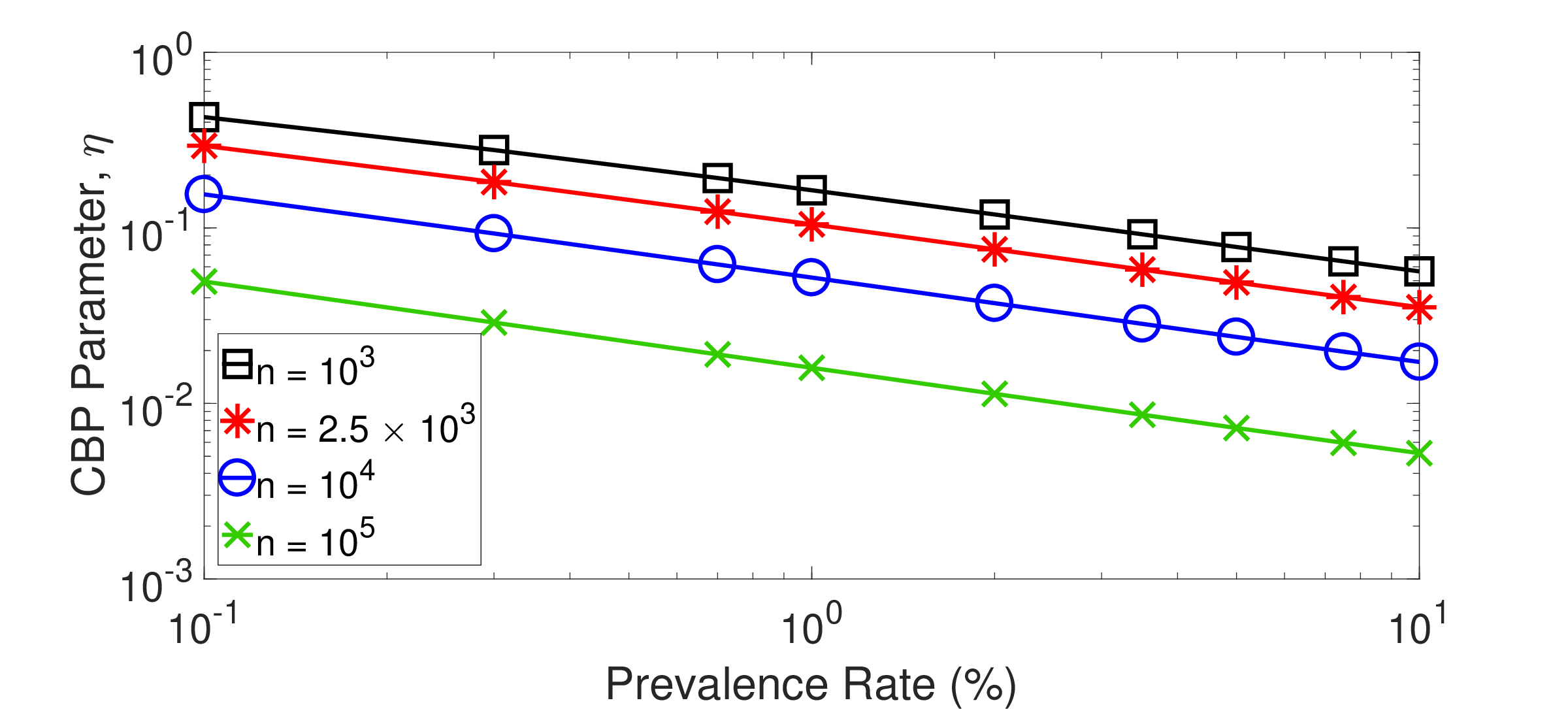}
	\caption{The CBP parameter, $\eta$, as the prevalence rate (\%) is varied, computed for different values of $n$, with $g_\epsilon = 0$, $\delta = 10^{-5}$, $c=1/2$, $s=s^*$.}
	\vspace{-0.1in}
	\label{fig:eta_vs_n}
\end{figure}

\subsection{Numerical Analysis of the DD Tuning Parameter, $\tilde{g}$}
\label{sec:DD Tuning Parameter}
We present an alternate view of  Theorem~\ref{thm:DD_Bound}, where the sufficient number of tests for DD algorithm to succeed with confidence $1-\delta$ when $\epsilon$ approximation errors are allowed is given implicitly by~\eqref{eq:DD_Bound}. We set $p = 1/k$ and relate $\epsilon$ to the allowed number of false negative errors, $d_\epsilon$, by~\eqref{eq:Prob_Error_FN_Relationship}. We start by posing the problem of solving for the sufficient number of tests in Theorem~\ref{thm:DD_Bound} as an optimization problem:
\begin{align}
	{m}_S, \tilde{g} &= \underset{m \in \mathbb{Z}_+,~\tilde{g}^\prime \in (0,\infty)}{\argmin}~m \nonumber \\
	&\text{s.t.}~\sum_{g=0}^{n\!-\!k} (1\!-\!(d_\epsilon\!+\!1)p(1\!-\!p)^{k\!-\!1\!+\!g})^m \mathbb{P}(G = g) \nonumber\\ &~~\leq (1-(d_\epsilon + 1)p(1-p)^{k-1+\bar{g}+\tilde{g}^\prime})^m \nonumber \\
	&~~~~~\text{and}~\nonumber\\	
	\label{eq:DD_gOpt}
	&{k \choose d_\epsilon\!+\!1}(1-(d_\epsilon + 1)p(1-p)^{k-1+\bar{g}+\tilde{g}^\prime})^m \leq \delta,
\end{align}
where $p = 1/k$, $d_\epsilon < k$ and $\mathbb{P}(G = g)$ is given by~\eqref{eq:DD_Prob_G}. We solve the optimization problem in~\eqref{eq:DD_gOpt} using a grid-search over a suitable range of $m$ and $\tilde{g}^\prime$ to ensure that we are operating in the feasible range.

\begin{figure}[t]
	\centering
	\includegraphics[width=0.8\columnwidth]{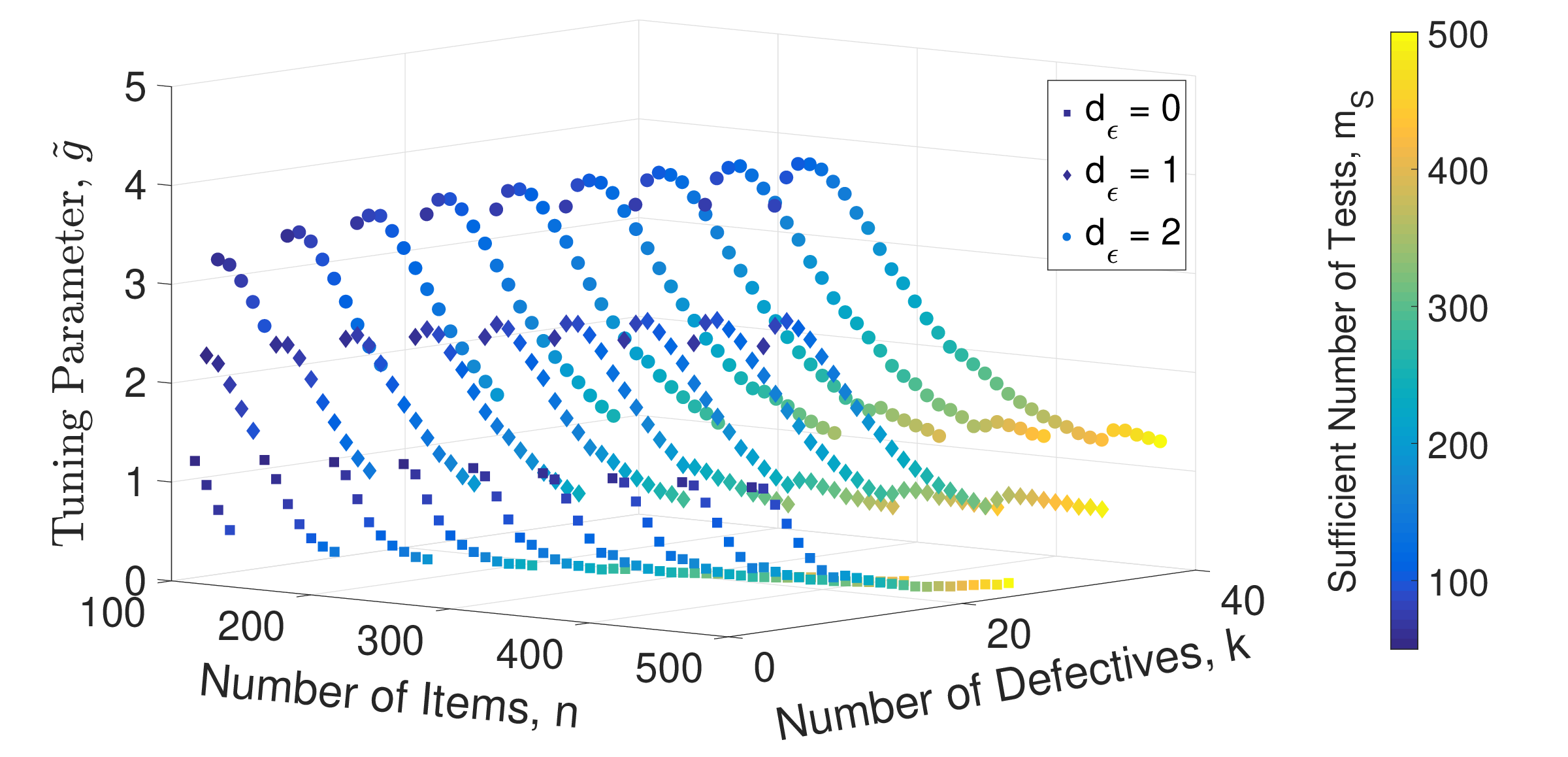}
	\caption{Solving for $\tilde{g}$ numerically by varying $n$ from $100$ to $500$ in steps of $50$ and suitable values of $k$ where the DD algorithm succeeds with $\delta = 0.01$ and $d_\epsilon \in \{0, 1, 2\}$.}
	\vspace{-0.1in}
	\label{fig:DD_g_tilde_grid_search}
\end{figure}

Figure~\ref{fig:DD_g_tilde_grid_search} shows a scatter plot of the optimum $\tilde{g}$ obtained by solving~\eqref{eq:DD_gOpt} with $p = 1/k$ by varying the population size, $n$, in range $[100,500]$ in steps of $50$ with suitable values of $k$, $\delta = 0.01$ and $d_\epsilon \in \{0,1,2\}$. The color bar in Fig.~\ref{fig:DD_g_tilde_grid_search} shows the sufficient number of tests, $m_S$, obtained by solving the optimization problem. We observe that $\tilde{g}$ increases with $d_\epsilon$ since the values of $g$ at which $\mathbb{P}(G = g)$ becomes negligible shifts to the right. Further, as $k$ increases, $\tilde{g}$ shows as slight increase, achieving a peak value at $k \approx d_\epsilon+4$ and then, tapering off as $k^{\alpha-1} e^{-\beta k} \beta^\alpha$ with $\alpha \approx 1.667$ and $\beta \in [1/5, 1/3]$. Finally, $\tilde{g}$ varies as $\log^2 n$ for a fixed~$k$.

\begin{figure*}[t]
	\centering
	\hspace*{-0.1in}
	\includegraphics[width=0.9\linewidth]{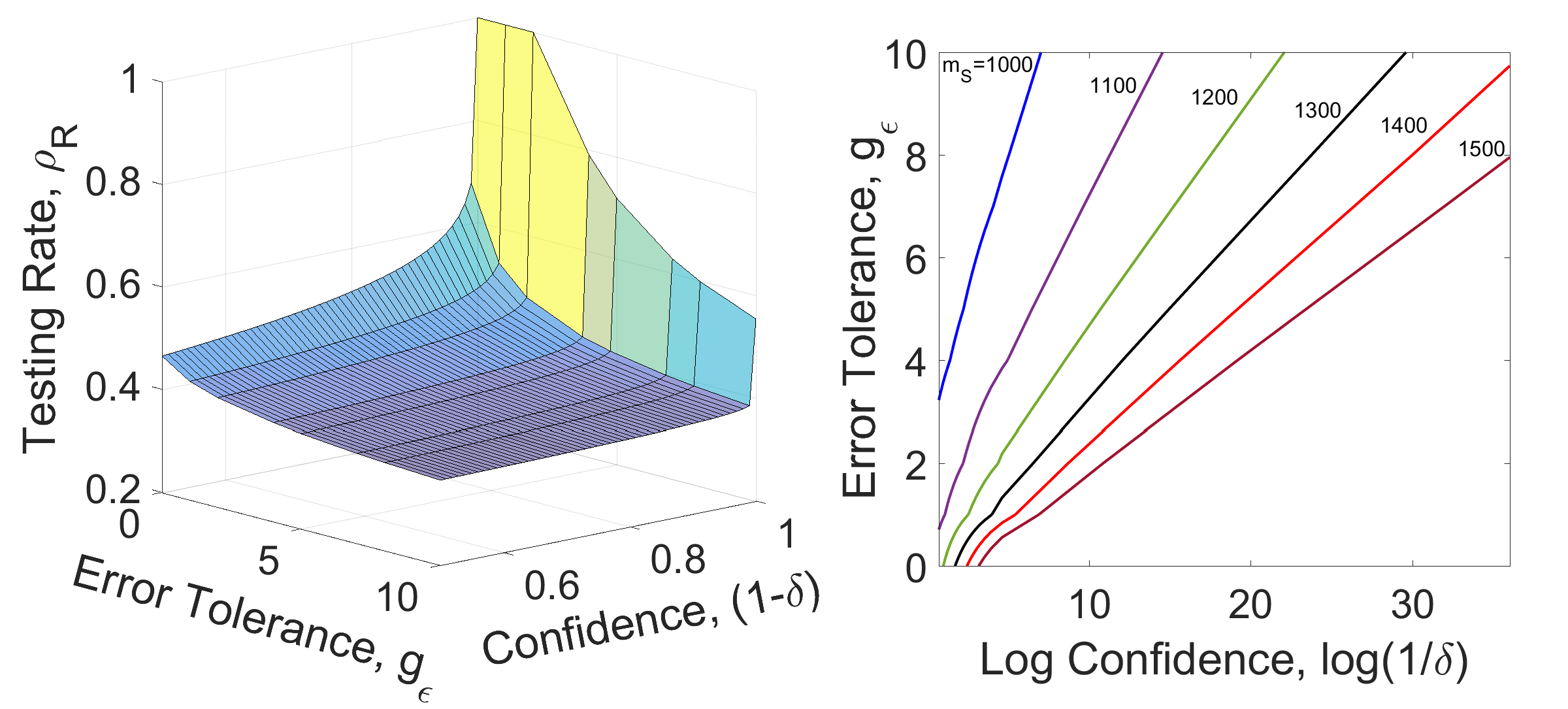}
	\vspace{-0.1in}
	\caption{(Left) The sufficient testing rate \emph{surface} and (Right) sufficient number of tests \emph{contours} vs. the confidence parameter, $1-\delta$, and error tolerance, $g_\epsilon$ with $n = 2500$, $k = 50$ and $p = 1/k$ for the CoMa algorithm.}
	\vspace{-0.1in}
	\label{fig:surface_contour_coma}
\end{figure*}
\begin{figure*}[t]
	\centering
	\hspace*{-0.1in}
	\includegraphics[width=0.9\linewidth]{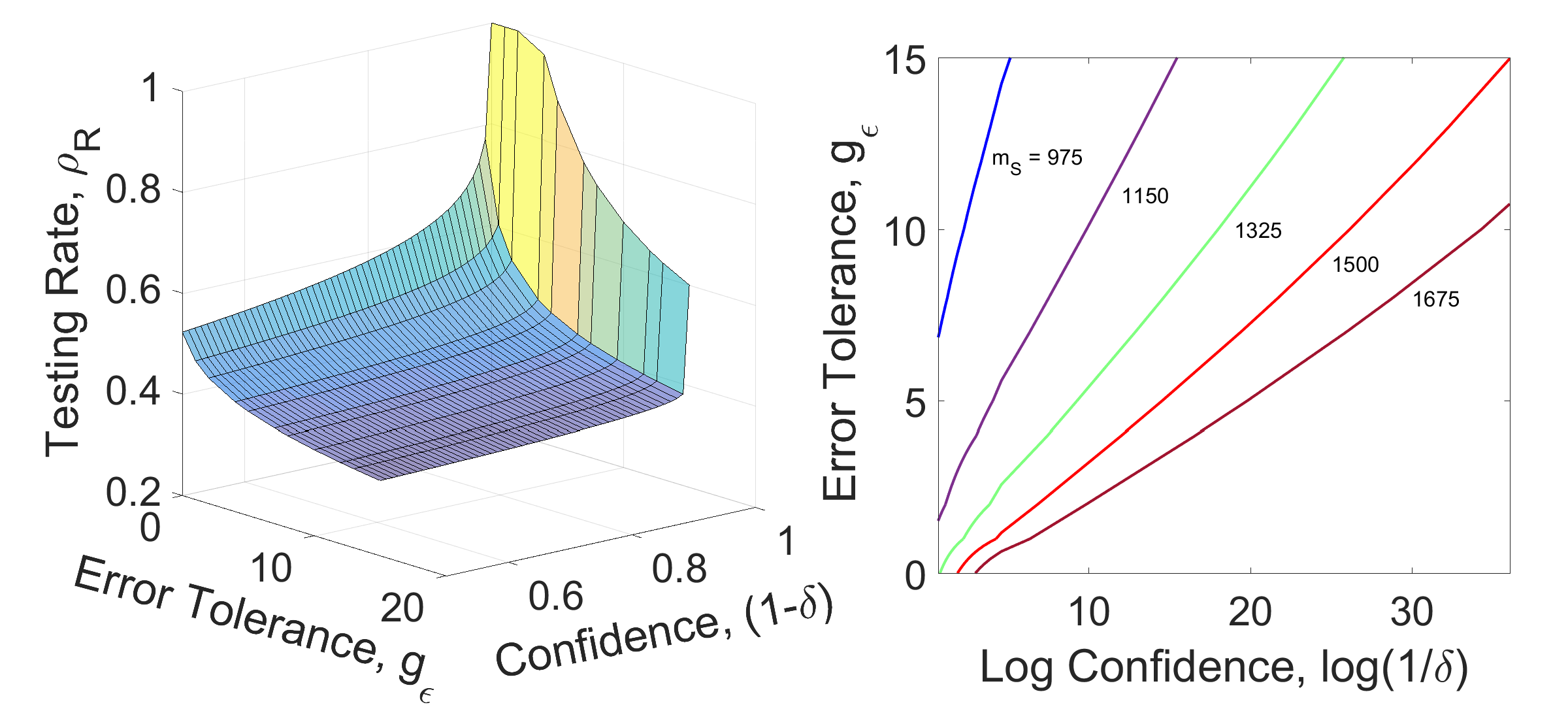}
	\vspace{-0.1in}
	\caption{(Left) The sufficient testing rate \emph{surface} and (Right) sufficient number of tests \emph{contours} vs. the confidence parameter, $1-\delta$, and error tolerance, $g_\epsilon$ with $n = 2500$, $k = 50$, $c = 1/2$ and $s = s^*$ for the CBP algorithm.}
	\vspace{-0.1in}
	\label{fig:surface_contour_cbp}
\end{figure*}
\begin{figure*}[t]
	\centering
	\hspace*{-0.1in}
	\includegraphics[width=0.9\linewidth]{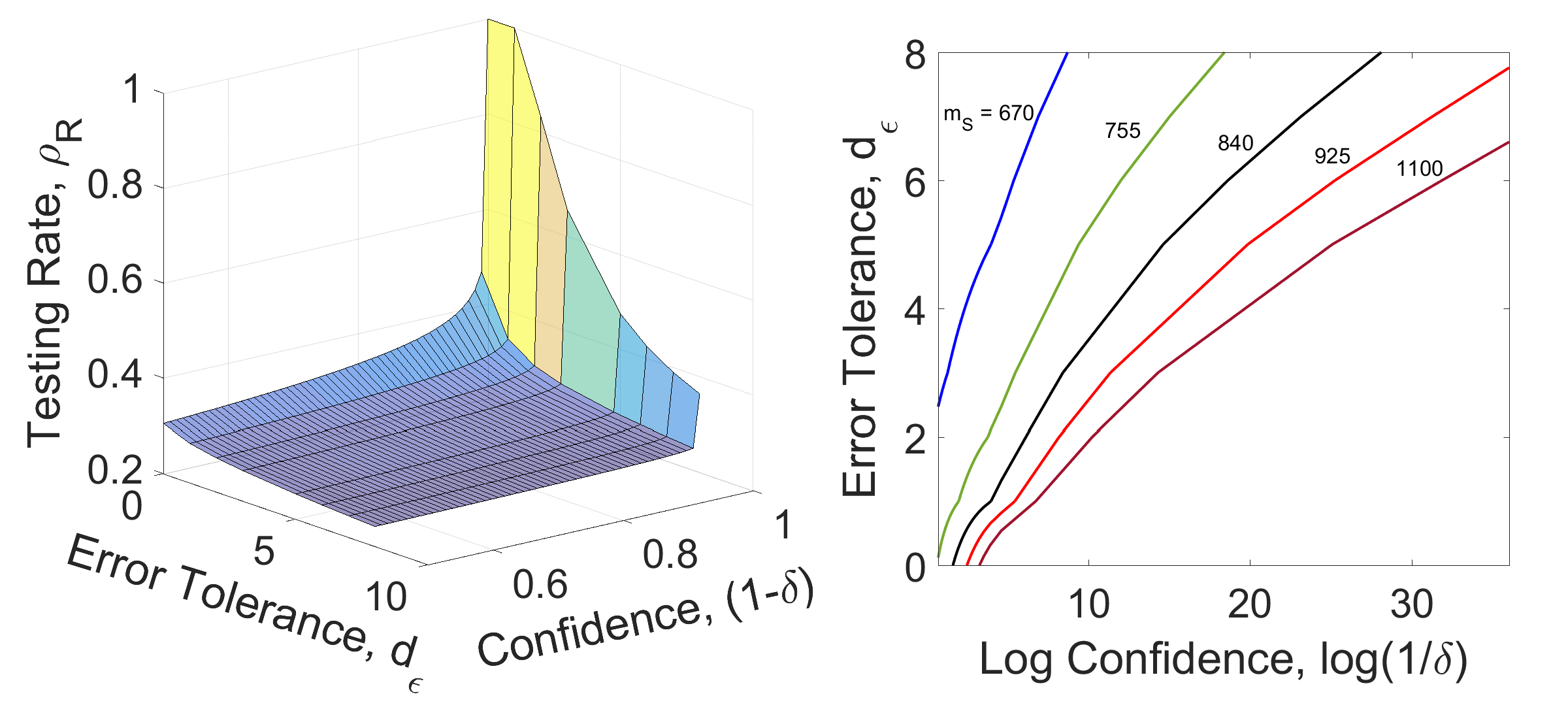}
	\vspace{-0.1in}
	\caption{(Left) The sufficient testing rate \emph{surface} and (Right) sufficient number of tests \emph{contours} vs. the confidence parameter, $1-\delta$, and error tolerance, $d_\epsilon$ with $n = 2500$, $k = 50$ and $p = 1/k$ for the DD algorithm.}
	\vspace{-0.1in}
	\label{fig:surface_contour_dd}
\end{figure*}

\subsection{Testing Rate Surface and Sufficient Tests Contours}
\label{sec:Testing Rate Surface and Sufficient Tests Contours}
We now introduce the notions of the \emph{testing rate surface} and \emph{sufficient tests contours}, which allow better visualization of the trade-off between the error margin and the confidence parameter. These are illustrated in the left (testing rate surface) and right (sufficient tests contour) subplots of Figs.~\ref{fig:surface_contour_coma},~\ref{fig:surface_contour_cbp} and~\ref{fig:surface_contour_dd} for CoMa, CBP, and DD algorithms, respectively. Specifically, the testing rate surface shows the sufficient testing rate, $\rho_R$, as a function of the error tolerance $g_\epsilon$ (or $d_\epsilon$) and confidence $1-\delta$. The testing rate contours are plotted over log confidence and mark the boundary over which a given number of tests are sufficient. For example, referring to Fig.~\ref{fig:surface_contour_coma}, all error tolerance and confidence values to the right and under the  blue curve are achievable when $1,\!000$ group tests are used. From the left subplots, we see that allowing for a nonzero error tolerance allows us to reach a high confidence level without significantly increasing the number of tests. On the other hand, if exact recovery is required, the number of tests rapidly increases as the confidence approaches one. Similar observations on the trade-off between the error tolerance and confidence parameter can be made from the right subplots also. Comparing the testing rate surfaces of CoMa, CBP, and DD algorithms, we note that although the geometry of the surfaces is visually similar, the testing rate offered by the DD algorithm is better by a factor $\approx 1.5$ compared to that of CoMa and by a factor $\approx 1.6$ compared to that of CBP when the confidence level is close to $1$, across various error tolerances $\epsilon$. In other words, at a given confidence level, allowing for a small number of missed defectives leads to a larger reduction the number of group tests compared to allowing for a small number of false positives. Thus, at high confidence ($\approx 1.0$) and low $\epsilon$, the CBP algorithm requires the highest number of tests at $1675$, followed closely by CoMa at $1500$ and then by a larger margin, by the DD algorithm at $1100$.

\subsection{Testing Rate vs. Population Size}
\label{sec:Testing Rate Curves vs. Population Size Across Various Sparsity}

\begin{figure}[t]
	\centering
	\includegraphics[width=0.8\columnwidth]{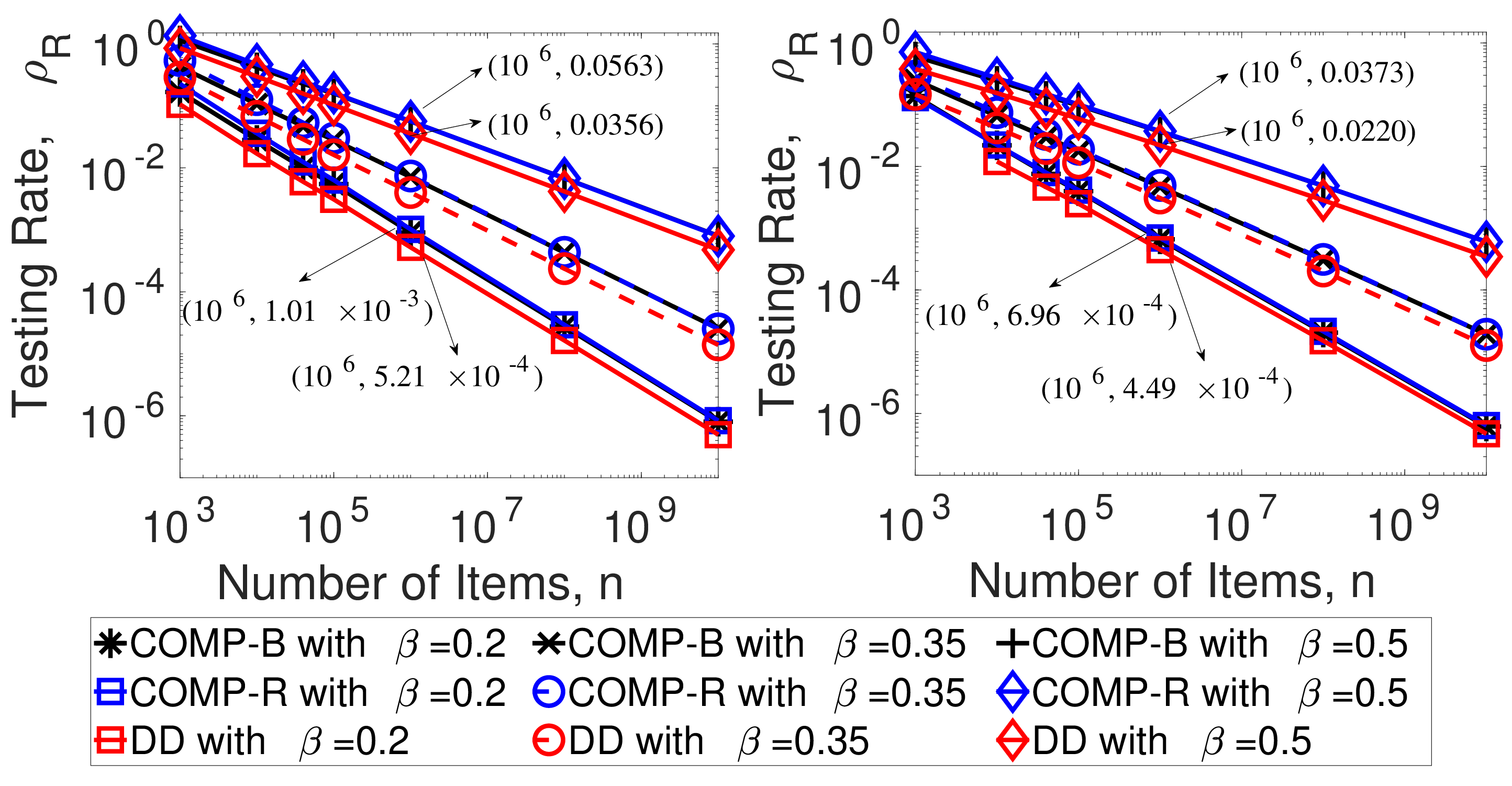}
	\caption{Sufficient testing rate vs. population size, $n$, across various inverse-sparsity parameter $\beta$ with $k = 0.95n^\beta$, $s=s^*$, $c = 1/2$, $p = 1/k$ and $\delta = 10^{-3}$ for CoMa, CBP and DD algorithm in exact recovery case in the left subplot and approximate recovery case ($g_\epsilon\!=\!d_\epsilon\!=\!5$) in the right subplot.}
	\vspace{-0.1in}
	\label{fig:testing_rate_n_sparsity}
\end{figure}
In the sub-linear regime, we choose $k = \Theta(n^{\beta})$, where $\beta \in (0,1)$ is called the \emph{inverse-sparsity} parameter, because the number of defective items increases (i.e., the item vector $\mathbf{x}$ becomes less sparse) as $\beta$ increases. In Fig.~\ref{fig:testing_rate_n_sparsity}, we compare the testing rates of the three algorithms as $n$ increases across various values of $\beta \in \{0.2, 0.35, 0.5\}$ at high confidence, i.e., when $\delta = 10^{-3}$. The left subplot in the figure shows the testing rate under  exact recovery. The right subplot shows the behavior for approximate recovery, i.e., with $g_\epsilon = d_\epsilon = 5$.

From the left subplot of Fig.~\ref{fig:testing_rate_n_sparsity}, we see that, for exact recovery, a testing rate of $\approx 0.0563$ and $1.01\times10^{-3}$ for CoMa, CBP and a testing rate of $\approx 0.0356$ and $5.21\times10^{-4}$ for DD at $\beta = 0.5$ and $0.2$, respectively, are sufficient when the population size is $n=10^6$. Similarly, from the right subplot, we see that the testing rates relax to $0.0373$ and $6.96\times10^{-4}$ for CoMa, CBP and to $0.022$ and $4.49\times10^{-4}$ for DD at $\beta = 0.5$ and $0.2$, respectively, at the same population size when we allow for $5$ errors. In summary, the testing rates relax by $\approx 33\%$ and $31\%$ for CoMa, CBP and by $\approx 38\%$ and $13\%$ for DD given $n = 10^6$ when we allow for $5$ errors at $\beta = 0.5$ and $0.2$, respectively. The lower percentage change in the case of the DD algorithm at $\beta = 0.2$ can be attributed to the dependency of the sufficient number of tests on $\log{k \choose d_\epsilon\!+\!1}$ in DD vs. that on $\log{n-k \choose g_\epsilon\!+\!1}$ in CoMa, CBP algorithms.

As $\beta$ increases, the testing rates of all three algorithms increase. More specifically, we have seen in Sec.~\ref{sec:The CoMa Algorithm} that $O\left(k\log n\right)$ tests are sufficient for exact recovery for large $n$. With $k = \Theta\left(n^\beta\right)$, 
we get the following: $\log \rho_R = O\left((\beta-1)\log n + \log(\log n)\right)$. Therefore, $\log \rho_R$ is approximately linearly decreasing with $\log n$ with slope $\beta-1$ for large $\log n$, which matches with the exact bound shown in the plot. From the curves in the right subplot, we see that our observation on the slope holds true even for the approximate recovery case. Finally, the DD algorithm performs the best in terms of the sufficient number of tests required, across all $\beta$s. While this is known for the exact recovery case~\cite{Aldridge_Balsassini_2014, Aldridge_GT_IT_2019}, we see from the right subplot that a similar observation holds for the approximate recovery case also.

\section{Discussion and Conclusion}
\label{sec:Discussion_and_Conclusion}
We studied three well-known practical boolean group testing algorithms: CoMa, DD, and CBP, from the perspective of the function learning framework. We used the probably approximately correct (PAC) analysis tools to derive a sufficiency bound on the number of tests needed for approximate defective set recovery for all three algorithms. We showed that the PAC-based bounds reduce to known sufficiency bounds in the case of exact recovery. Also, we presented numerical results illustrating the bounds vis-\`{a}-vis the desired confidence level, for both exact and approximate recovery cases. Finally, we demonstrated the advantage of the PAC formulation by empirically illustrating its ability to quantify the sufficient number of group tests across various confidence levels and error tolerances. 

Future work could involve theoretical analysis of the group testing algorithms under different test matrix formulations like (near-) constant column-weight and doubly-regular design, other group testing algorithms like the smallest satisfying set (SSS) and linear programming (LiPo) decoder, and extending the PAC framework to noisy group testing scenarios. The PAC formulation presented in this work is different from the classical PAC analysis since the data distribution can be chosen based on the hypothesis space containing the target function. This aspect could lead to novel results on the existence and characterization of the algorithms that are PAC-learnable. This is also an interesting direction for the future work.

\nocite{*}
\bibliographystyle{IEEEtran}
\bibliography{IEEEabrv,References.bib}

\end{document}